\documentclass[11pt]{article}
\usepackage{graphicx} % Required for inserting images
\usepackage[utf8]{inputenc}
\usepackage[margin=1in]{geometry}
\usepackage{amsmath,amsthm,amssymb}
\usepackage{xcolor}
\usepackage[hyphens]{url}
\usepackage{hyperref}
\usepackage[nameinlink,capitalise]{cleveref}
\usepackage{bbold}
\usepackage{mathtools}
\usepackage{doi}
\usepackage{tikz}
\usepackage{thm-restate}
\usepackage[ruled]{algorithm2e}
\usepackage[T1]{fontenc}
\usepackage{indentfirst}
\usepackage{libertine}

\usepackage{etoolbox}
\usepackage{todonotes}
\newtoggle{DEBUG}
\togglefalse{DEBUG}
\newtoggle{COMMENTS}
\toggletrue{COMMENTS}

\usetikzlibrary{arrows.meta, positioning}

\newcommand{\E}{\mathbb{E}}

\renewcommand{\P}{\mathcal{P}}

\newcommand{\fui}{\textsc{Full Revelation}}
\newcommand{\ber}{\mathtt{Bernoulli}}
\newcommand{\sgm}{\textsc{Single Mean}}

\newcommand{\iss}{independent signaling scheme}

\newcommand{\ISS}{Independent Signaling Scheme}
\newcommand{\sgp}{signaling policy}
\newcommand{\sgps}{signaling policies}

\newcommand{\SGP}{Signaling Policy}
\newcommand{\SGPs}{Signaling Policies}
\newcommand{\mpr}{mapping rule}
\newcommand{\slr}{selection rule}

\newtheorem{theorem}{Theorem}[section]
\newtheorem{lemma}[theorem]{Lemma}

\newtheorem{proposition}[theorem]{Proposition}
\theoremstyle{definition}
\newtheorem{definition}[theorem]{Definition}
\newtheorem{example}[theorem]{Example}
\newtheorem*{remark}{Remark}

\DeclareMathOperator*{\argmax}{arg\,max}

\renewcommand{\d}{\,\mathrm{d}}

\newcommand{\pb}{\mu}
\newcommand{\pw}{w}

\crefname{Program}{Program}{Programs}
\creflabelformat{Program}{(#2\textup{#1})#3}

\renewcommand{\epsilon}{\varepsilon}

\makeatletter
\def\@fnsymbol#1{\ensuremath{\ifcase#1\or \dagger\or \ddagger\or
   \mathsection\or \mathparagraph\or \|\or **\or \dagger\dagger
   \or \ddagger\ddagger \else\@ctrerr\fi}}
\makeatother
\newcommand*\samethanks[1][\value{footnote}]{\footnotemark[#1]}

\definecolor{linkc}{rgb}{0.1, 0.5, 0.7}
\definecolor{citec}{rgb}{0.6, 0.3, 0.7}
\definecolor{urlc}{rgb}{0.5, 0.1, 0.2}
\hypersetup{
    colorlinks=true,
    linkcolor=linkc,
    citecolor=citec,
    urlcolor=urlc
}

\title{Majorized Bayesian Persuasion and Fair Selection}
\author{Siddhartha Banerjee\thanks{School of Operations Research and Information Engineering, Cornell University, Ithaca, NY 14853. Email: \textsf{sbanerjee@cornell.edu}. Supported by NSF grants ECCS-1847393, CNS-195599 and AFOSR grant FA9550-23-1-0068.} \and Kamesh Munagala\thanks{Department of Computer Science, Duke University, Durham, NC 27708-0129. Emails: \textsf{kamesh@cs.duke.edu}, \textsf{yiheng.shen@duke.edu}. Supported by NSF grants CCF-2113798 and IIS-2402823.} \and Yiheng Shen\samethanks[2] \and Kangning Wang\thanks{Department of Computer Science, Rutgers University, Piscataway, NJ 08854. Email: \textsf{kn.w@rutgers.edu}.}}
\date{}

\begin{document}

\maketitle

\begin{abstract}
We address the fundamental problem of selection under uncertainty by modeling it from the perspective of Bayesian persuasion. In our model, a decision maker with imperfect information always selects the option with the highest expected value. We seek to achieve fairness among the options by revealing additional information to the decision maker and hence influencing its subsequent selection. To measure fairness, we adopt the notion of majorization, aiming at simultaneously approximately maximizing all symmetric, monotone, concave functions over the utilities of the options. As our main result, we design a novel information revelation policy that achieves a logarithmic-approximation to majorization in polynomial time. On the other hand, no policy, regardless of its running time, can achieve a constant-approximation to majorization. Our work is the first non-trivial majorization result in the Bayesian persuasion literature with multi-dimensional information sets.
\end{abstract}

%\begin{abstract}
%We consider the fundamental problem of selection under uncertainty, and model it as a canonical special case of Bayesian persuasion. We assume a simple and explainable decision maker who always selects the option with the highest expected value given its information about the options. The goal is to design a randomized policy for revealing additional information about the options to the decision maker, so that when it subsequently makes a selection decision, this is fair to the options. We measure fairness in terms of majorization, and seek to devise an information revelation policy that simultaneously approximately maximizes all symmetric, monotone, concave functions over the utilities of the options. Our main result is a logarithmic-approximation to majorization via a polynomial-time-computable information revelation policy, and we show that a constant-approximation cannot be achieved regardless of running time. Our work is the first non-trivial result for majorization in Bayesian persuasion with multi-dimensional information sets.
%\end{abstract}

\vspace{4pt}
\makeatletter
\renewcommand\tableofcontents{%
    \@starttoc{toc}%
}
\makeatother
\tableofcontents

\thispagestyle{empty}
\newpage
\setcounter{page}{1}
\section{Introduction}
%\textcolor{orange}{TODO: mention computational efficiency in intro.}

Selection problems arise in hiring decisions, college admissions, ad auctions, and many other scenarios. Typically, a selection process involves defining an \emph{explainable} selection rule that solves an optimization problem over the agents (e.g., applicants or advertisers). Perhaps the simplest and most common rule is to choose the best or most qualified agent(s) -- for example, many countries use standardized test scores for college admissions. Under natural simplifying assumptions, such a meritocratic rule will maximize the utilitarian social welfare, and can be perceived as fair. Nevertheless, its outcomes may not always align with other fairness considerations. For instance, one can adopt the view of max-min fairness and aim at helping the worst-off agent.

The meritocratic rule has a perhaps more subtle issue: when trying to evaluate the agents and define the ``best'' one, the quality measurements may be imperfect. Indeed, the decision maker uses observable information or features of the agents to reach a decision, but counterfactually, a different decision could have been made in the presence of more certain or complete information. This issue can have disproportionate negative impacts on certain demographic groups and therefore undermine fairness perceptions towards the rule. To address the issue, there has been a long line of recent research focusing on modeling this uncertainty and constructing randomized selection policies to achieve fair outcomes~\cite{kleinberg2018selection,celis2020interventions,singh2021fairness,shen2023fairness,devic2024stability}. As a practical example, in many hiring situations, the interviewers are required to interview at least a certain number of candidates from prespecified demographic groups, as an attempt to mitigate the impact of such uncertainty and to achieve fairness.

%Selection problems arise in many settings, including hiring decisions, college admissions, ad auctions, and so on. Typically, selection involves solving an optimization problem over the agents (whether it be applicants or ads), and the usual goal is to choose a selection rule that is {\em explainable}. The simplest and most widely used rule is to choose the best or most qualified agent. Such a rule will maximize social welfare, and will clearly be perceived as fair. 

%There is however a gaping and widely known issue with any such rule, and that has been much of a focus of a long line of recent work~\cite{kleinberg2018selection,celis2020interventions,singh2021fairness,shen2023fairness,devic2024stability}. The decision maker uses information or features of the agents -- think resumes for job applicants, or transcripts, essays, etc for college applicants, and these give incomplete or uncertain information about the applicants. Indeed, for any selection decision that is made in the presence of such uncertain or incomplete information, there is a counterfactual decision that could have been reached in the presence of more certain or complete information.  Much of the recent work cited above constructs randomized selection policies to achieve a fair outcome among the applicants given such uncertainty in quality. Even in practice, the celebrated Rooney rule~\cite{kleinberg2018selection} used by the NFL stipulates interviewing at least one minority candidate, and attempts to mitigate the impact of such uncertainty and achieve fairness.

\paragraph{Fairness via Information Revelation.} We diverge from this line of research and take a view inspired by the influential Bayesian persuasion literature, starting with the seminal work of Kamenica and Gentzkow~\cite{Kamenica2011bayesian}. We ask:

\begin{quote}
Given uncertainty about agent qualities, how much additional information should be gleaned and revealed to the decision maker (either by the agents themselves or by a central entity), so that even if the decision maker sticks to ``selecting the best'', it still yields a fair outcome to the agents?
\end{quote}

To approach this high-level question, we use the following model. (See Au and Kawai \cite{AuKawai} for a related model.) Each agent has an uncertain scalar quality, and its true quality is known to an entity -- either the agent itself or an intermediary. This entity then signals additional information independently for each agent to the decision maker to refine its uncertainty via Bayes' rule. Analogous to~\cite{AuKawai}, this signaling process can be done in a distributed fashion by the agents themselves, a property that is often desirable for the applications we consider in order to preserve privacy of agents' true quality. Subsequently, it can happen that several agents have comparable posterior quality in the eyes of the decision maker, and a fair selection can be made by the decision maker without significantly sacrificing its own optimality notion of selecting the posterior best. 

As a concrete motivation, consider designing a standardized form for admission or hiring. The form is the signaling scheme - it is designed centrally by the decision maker, filled separately by agents, and the decision maker can break ties based on information revealed in the forms, sometimes in conjunction with a lottery number assigned to each agent, to ensure fairness in selection. We present the formal model in \cref{sec:prelim}, and an example in \cref{eg:maxmin}. 

In order to define fairness, we consider the vector of expected utilities received by the agents, where each agent receives utility equal to its true value if it is selected. We seek signaling (or information revelation) policies whose utility vector is \emph{approximately majorized}~\cite{hardy1934inequalities,goel2005approximate,kumar2006fairness,chakrabarty2019approximation}, meaning that all symmetric, monotone, concave functions over the utilities are simultaneously approximately maximized by the same policy. Such fairness functions capture, for instance, max-min fairness (maximizing the minimum utility) and the Nash welfare (maximizing the geometric mean of the utilities).

At one extreme, if the agents do not reveal any additional information, the agent with the highest expected quality, according to the decision maker's prior, is selected. As mentioned before, this can be problematic in terms of fairness to counterfactual information. We show in \cref{ex:full2} that the other extreme -- where each agent reveals its exact quality (\fui{}) -- may be far from fair as well, even when the qualities only take ``high'', ``medium'', and ``low'' values. This example highlights the need to carefully choose which information to reveal.

We note that such information revelation considerations extend beyond the motivating settings described above. As an example, consider a government agency (the decision maker) that has fixed rules for allocating money designated to a social welfare cause such as refugee resettlement. Local agencies typically are better informed about welfare needs of the groups they serve, and can selectively reveal information in order to facilitate fairness for these groups. 

\paragraph{Summary of Results.} 
We present a detailed summary of our results in \cref{sec:prelim} after we present our formal model. Informally, our main result is a novel set of signaling (or information revelation) mechanisms that achieve approximate majorization. (See \cref{thm:approx}.) The key technical hurdle with designing such policies is the behavior of the decision maker -- this decision maker chooses the best (or approximately best) agent given its information. For any given fairness objective (such as max-min fairness), a general algorithmic result due to Dughmi and Xu~\cite{dughmi2016algorithmic} yields an FPTAS for any Bayesian persuasion problem with arbitrary information sets, assuming that the sender can correlate agent signals. %This essentially reduces the problem to multi-dimensional mechanism design~\cite{CaiDW12,BhalgatGM13} via writing an exponential size LP capturing the decision maker's behavior and solving it via sampling. 
However, such a generic approach does not shed light on the existence of policies that \emph{simultaneously} approximate any fairness objective, and this fact requires us to unearth specific structure in our problem.
%\km{Emphasize that we get an existential bound. }

%Typically, majorization results require a succinct characterization of the space of solutions, say via a compact mathematical program that does not enu. The hurdle with such a formulation is that there is a non-linear interaction between two types of variables -- those encoding mapping agent qualities to signals, and those encoding the decision maker's selection behavior given the signal. This interaction is not just non-linear, but non-convex as well. Such behavior is typical in Bayesian persuasion settings where decision makers of fixed behavior need to be {\em persuaded} to act in a specific fashion. 

Our main contribution in \cref{sec:approx} is in unearthing such structure for the selection problem via a novel class of signaling policies with simple structure, leading to positive {\em existence} and {\em computational} results for approximately majorized policies. We term the novel set of policies as \sgm{}, and these generate, in a randomized fashion, information signals where the posterior inferred by the decision maker maps the agents' values to a common mean quality. This mapping to the common mean happens with as large a probability as possible, hence giving the decision maker wide leeway in implementing a fair selection even when it ``goes with the best''. Though this policy sounds intuitive, its analysis is far from obvious as we discuss below. %since agents can generate not only the common mean signal, but also other signals depending on their underlying quality, leading to a non-convex optimization problem. 
%We need a key structural insight that decouples the mapping variables in the mathematical program from the selection variables, and a subsequent insight that revealing full information is indeed exactly majorized for simple Bernoulli distributions on quality. 
Informally, the final result is a polynomial-time-computable and $O(\log V)$-approximately majorized policy\footnote{Our results are bicriteria, and assume the societal decision maker is a $(1+\epsilon)$-approximate welfare maximizer.}, where $V$ is the ratio of the largest to smallest quality value. This means for any fairness function, the policy yields an $O(\log V)$-approximation in polynomial time.  

As a corollary, our work also presents an approximation algorithm for any given fairness function when agents generate signals independently of other agents. As shown in \cref{sec:asym}, our results easily extend to handle the case where the agents' utility (on which we seek to be fair) is different from their quality (whose posterior is used by the social planner to perform allocation). We note that the FPTAS for general Bayesian persuasion in~\cite{dughmi2016algorithmic} assumes a more powerful intermediary that can send a common signal using the information of all agents.\footnote{They also present an independent signaling scheme when the actions have i.i.d. rewards. In our case, the i.i.d. setting is uninteresting since revealing full information about agent quality is trivially $1$-majorized.} Our work shows that for the weaker and practically motivated  intermediary that generates per-agent signals independently of other agents, there is still a polynomial-time-computable $O(\log V)$-approximation. %(see \cref{cor1}).

In \cref{sec:main_lb}, we complement this positive result by showing that a mild dependence on $V$ is unavoidable -- no signaling scheme can be $(\log \log V) / 3$-approximately majorized.

\paragraph{Technical Highlight: Reduction to Majorized Network Flow.} At a technical level, our majorization results proceed by reducing the problem to network flow with a source and multiple sinks, for which a majorized solution has long been known~\cite{veinott1971least,megiddo1974optimal}. To enable such a reduction, we need linear programming relaxations from the perspective of the signaling entity. Such a formulation is novel and non-trivial since  a natural mathematical program from the perspective of the signaling entity is non-convex. This is because any program needs to encode the decision maker's behavior. Handling this inherent non-convexity necessitates new structural insights (\cref{lem:struct}). Our LP formulation connects the signaling problem to classic stochastic optimization problems, such as stochastic knapsack~\cite{DBLP:journals/mor/DeanGV08} and multi-armed bandits~\cite{DBLP:journals/jacm/GuhaMS10}. This is in contrast to the more traditional and generic LP formulation for Bayesian persuasion~\cite{dughmi2016algorithmic,cummings2020algorithmic} that directly encodes the optimal behavior of the decision maker, reducing the problem to multi-dimensional mechanism design~\cite{cai2012optimal}. 

To build up to the general approximate majorization result, we first consider the simpler class of \fui{} policies alluded to above that reveal all information to the decision maker. Though such policies cannot be approximately majorized in general (see \cref{ex:full2}), we show in \cref{sec:network_flow} that when agent information is Bernoulli, such policies are without loss of generality, and are exactly majorized (hence best possible in terms of fairness). For arbitrary distributions, we show that there is always a $2$-approximately majorized policy within the class of \fui{} policies (see \cref{sec:full_maj}). These results showcase the network flow approach, which may be of independent interest in devising fair policies for other signaling problems. %This majorization result will form a key component in our general approximation result described above.

\subsection{Related Work}
\paragraph{Bayesian Persuasion.} Our selection model is a special case of \emph{information design} (see \cite{bergemann2019information,dughmi2017algorithmic} for surveys) where an information mediator provides information to impact the behavior of one or more decision makers. This has also been termed \emph{signaling} or \emph{persuasion} in the literature. In the original model of \emph{Bayesian persuasion} by \cite{Kamenica2011bayesian}, there is one agent called the \emph{receiver} who receives additional information from a better-informed \emph{sender}. Given the signal, the receiver computes its posterior over the state of nature and chooses an action to maximize its own utility. The sender can design the signals so that the receiver, acting in its own interest, maximizes some utility function the sender cares about. This problem has been widely studied in various contexts, such as price discrimination~\cite{bergemann2015limits,Banerjee2024fair}, security games~\cite{DBLP:conf/aaai/XuRDT15}, regret minimization~\cite{DBLP:conf/sigecom/BabichenkoTXZ21}, and other economic settings~\cite{DBLP:journals/mktsci/ChakrabortyH14}. 

%In our model for fair selection, the state of nature is the quality of all agents, and the receiver's utility maximizes expected quality given its information. The sender's utility is captured by a fairness function over the receiver's decision, using the same quality measure.

As mentioned earlier, computational approaches to the general Bayesian persuasion problem have been proposed, with an FPTAS for any given objective function~\cite{dughmi2016persuasion,dughmi2016algorithmic}. However, our goal of majorization requires simultaneous optimization for many objectives, and here, even showing an existence result requires new ideas. Our techniques, starting with our LP relaxations from the perspective of the signaling entity, are entirely different. As a consequence, unlike~\cite{dughmi2016algorithmic}, our methods work when the signaling scheme is independent across agents, while the computational results in~\cite{dughmi2016algorithmic} require correlation between the signals. 

\paragraph{Majorization.} The concept of majorized vectors is a classical one~\cite{karamata1932inegalite,hardy1934inequalities}, and is equivalent to solutions that simultaneously maximize symmetric concave functions of the coordinates. In the context of resource allocation and routing problems, an approximate version of this concept was defined by~\cite{goel2005approximate,goel2006simultaneous}; see also~\cite{kumar2006fairness,chakrabarty2019approximation}. It was shown by~\cite{goel2005approximate} that the best approximation factor is the solution of a linear program. However, the approximation factor is problem-dependent and can be linear in the number of dimensions. One major exception is the single-source multi-sink network flow problem, for which the elegant result of Veinott~\cite{veinott1971least} shows exact majorization. Our main contribution is showing a surprising connection between this work and the  selection problem. We show that the approximation factor for our problem is only logarithmic in the range of the quality scores, hence achieving stronger fairness properties than what the generic bounds would indicate. 

We note that the work of~\cite{Banerjee2024fair} presents a constant-approximation to majorization for the special case of price discrimination. However, the state of nature in price discrimination is single dimensional (value of a single buyer), while it is multi-dimensional for our setting (joint quality distribution for $n$ agents). Indeed, a constant-approximation is provably unachievable in our setting, and we need new technical ideas to show our approximation results. As far as we are aware, our work gives the first majorization results for persuasion problems with multi-dimensional type spaces.

\paragraph{Fair Selection.}
The meritocratic policy of selecting the candidate(s) with the highest evaluation scores is particularly natural and widely used. However, this policy can be at odds with fairness considerations. For instance, in a college admissions setting, standardized test outcomes can have correlations with demographic factors~\cite{dixon2013race}, and these correlations can undesirably and significantly skew demographic outcomes. To mitigate this, one existing approach is to have a quota system setting aside a number of spots for candidates from certain demographic groups. A prominent example is the Rooney Rule, which among other instructions, requires that each National Football League (NFL) team must ``employ a female or minority coach as an offensive assistant''.\footnote{ \url{https://operations.nfl.com/inside-football-ops/inclusion/inclusive-hiring}, accessed June 28, 2024.} Different versions of the Rooney Rule have been seen in the corporate world, in governments, and in academia -- those rules require employers to hire or at least interview candidates from prespecified demographic groups. Researchers have been studying the effectiveness of Rooney Rule variants in fair selection % and mitigating implicit bias in hiring
and hiring~\cite{kleinberg2018selection,celis2020interventions}.

Policies that consider demographic factors can fail when demographic information is unavailable, or when there is no consensus on which demographic groups should be protected. Interventions that explicitly consider demographic factors may also face societal and legal challenges. For instance, some practices have been deemed unlawful in U.S. Supreme Court decisions, such as racial quotas (Regents of the University of California v.\@ Bakke (1978)) and more recently, race-based affirmative action (Students for Fair Admissions v.\@ Harvard (2023)).
%In the context of college admissions, in the United States, racial quota systems were deemed unlawful in 1978 in the landmark decision \emph{Regents of the University of California v.\@ Bakke}.\footnote{Besides the quota system, many other forms of race-based affirmative action in college admissions have also become unlawful since the 2022 landmark decision \emph{Students for Fair Admissions v.\@ Harvard}.}  There, Justice Powell wrote that ``[r]acial and ethnic classifications of any sort are inherently suspect'', which perhaps, even to this day, hints at the challenges of designing fair selection policies. 
Additionally, revealing demographic information can have adverse impact in hiring situations on, for example, women~\cite{goldin2000orchestrating} and African Americans~\cite{bertrand2004emily}.
%due to so-called ``implicit bias'' %Our work provides a fairness notion via selective information revelation and majorization that does not explicitly consider demographic information.

Departing from the above approaches, recent work explicitly models uncertainty in evaluation measures~\cite{emelianov2020fair,mehrotra2021mitigating,garciasoriano2021maxmin,singh2021fairness,shen2023fairness,devic2024stability} and proposes fair randomized selection rules under models of fairness and uncertainty over applicant qualities and attributes. Our work falls in this framework, but differs in the following way: We incorporate information revelation about applicant qualities, in order to (Bayesian-)persuade the meritocratic decision maker to act in line with fairness considerations. 

\paragraph{Selection via Persuasion.} Our model is similar to the selection models in~\cite{AuKawai,Du2024}. Unlike our model, in their model, there is no intermediary, and the agents are free to choose their signaling schemes. They show that the resulting game over the agents has an equilibrium under fairly general conditions. In contrast, we consider the setting with an intermediary that seeks fairness across agents; in other words, we assume agents can coordinate their signaling scheme for mutual benefit. Further, the models in \cite{AuKawai,Du2024} assume the agent utilities are different from the social planner's. In particular, agent utilities are binary -- $1$ for being selected and $0$ otherwise -- while we assume they are the same as their quality. Our setting models how the intermediary or social planner will perceive utility (as the quality), while in their setting, the utility is how it would have been perceived by selfish agents. We note however that our results are robust to the choice of utilities, and as we show in \cref{sec:asym}, it is easy to modify our positive results to work as is in their utility model as well.

%\paragraph{Fair Selection and Implicit Bias.} College admissions and hiring decisions in general suffer from the pervasive problems of both incomplete information and implicit bias. For instance, standardized test outcomes are biased by demographic factors~\cite{dixon2013race}, while more holistic assessments can be impacted by implicit bias against marginalized groups~\cite{bertrand2004emily}. Deterministic selection in the presence of such uncertainty and bias can lead to pervasive discrimination against certain groups.

%The main approach to addressing bias without explicitly modeling uncertainty adopts the so-called {\em Rooney Rule}~\cite{kleinberg2018selection,celis2020interventions}, which reserves a number of spots to prioritize  minority applicants based on applicant features. However, these methods can fail when there is no pre-specified minority group, or when there are multiple notions of what ``minority'' can be. To address this, recent work explicitly models uncertainty in quality measures~\cite{emelianov2020fair,mehrotra2021mitigating,garciasoriano2021maxmin,singh2021fairness,shen2023fairness,devic2024stability} and proposes randomized selection rules that are fair under some model of uncertainty over applicant quality and attributes.  Our work falls in this framework, but differs in incorporating information revelation about applicant quality via persuasion to coax the decision maker to act in a fair fashion.

\section{Our Model and Results}
\label{sec:prelim}
A social planner wishes to select one agent from a set $A$ of $n$ agents $\{1, 2, \ldots, n\}$. Each agent $i$ obtains value $v_i$ if selected, where $v_i$ is drawn from a distribution $D_i$. We assume that the distributions $\{D_i\}_{i = 1}^n$ are mutually independent. In keeping with the Bayesian persuasion literature, we call the social planner the \emph{receiver}. The receiver knows the distributions $\{D_i\}_{i = 1}^n$, but not the values themselves, which are private information to the agents.

%We have a set $A$ of $n$ agents $\{a_1, a_2, \ldots, a_n\}$. A social planner wishes to select one agent. Each agent $a_i$ has value $v_i$ if selected, where $v_i$ is drawn from an independent distribution $D_i$. In keeping with Bayesian persuasion literature, we call the social planner as the {\em receiver}.  The receiver knows the distributions $\{D_i\}$, but not the values themselves, which are private to the agents. 

There is an information intermediary who knows the exact values $\{v_i\}_{i = 1}^n$. Again following the Bayesian persuasion literature, we call the information intermediary the \emph{sender}. The sender can reveal partial information about the values to the receiver via a \emph{\sgp}. (We will define \sgps{} and signals in \cref{sec:signal_def}.) The receiver hence receives a \emph{signal} $\sigma$, and uses $\sigma$ to update the prior value distributions $\{D_i\}_{i = 1}^n$ to posterior distributions $\{D_i(\sigma)\}_{i = 1}^n$. 

%There is an information intermediary who knows the exact value of each agent. Again following Bayesian persuasion literature, we call this the {\em sender}. This sender can reveal partial information about the agent values to the receiver via a {\em signaling scheme}. The receiver uses the signal $\sigma$ to update the prior $D_i$ of agent $i$ to a posterior distribution $D_i(\sigma)$. 

\subsection{Signaling Model}
%\subsection{Receiver Behavior: Welfare Maximization} 
\label{sec:signal_def}
% In our model, the receiver is a utilitarian welfare maximizer. Therefore, after the receiver updates the posterior of the agents' values to $\{D_i(\sigma)\}_{i=1}^n$, she will select an agent with the largest posterior mean -- that is, an agent in the set $S(\sigma) := \argmax_i \E[D_i(\sigma)]$. In the case of tie-breaking (i.e., when $|S(\sigma)| > 1$), we assume that the sender can tell the receiver who to select within the set $S(\sigma)$ via a \slr{}.
% Since the receiver is indifferent among all agents in this set $S(\sigma)$, the sender can choose the agent which the receiver breaks the tie in favor of.

%We assume the receiver is a welfare maximizer, so that she selects an agent with largest posterior mean, that is to $\mbox{argmax}_i \E[D_i(\sigma)]$.  We denote by $S(\sigma)$ the set of agents satisfying this property.

%In more generality, for $\theta \ge 1$, the receiver can be a $\theta$-welfare maximizer defined as follows.

%\begin{definition} [$\theta$-welfare maximizer]
%\label{def:relax}
%Given signal $\sigma$, let $\mu^*(\sigma) = \max_i \E[D_i(\sigma)]$. We say that a receiver is a $\theta$-welfare maximizer if she selects an agent $i$ with $\E[D_i(\sigma)] \ge \mu^*(\sigma)/\theta$.  We denote by $S(\sigma,\theta)$ the set of agents satisfying this property.
%\end{definition}

\paragraph{\ISS{}s.}
An \emph{\iss{}} $\omega$ works with a set $\Gamma$ of signals, and we will often use $\sigma \in \Gamma$ to denote a signal. An \iss{} in our model has two components: the \emph{\mpr} and the \emph{\slr}.

% \begin{itemize}
% \item The \emph{\mpr} is a function that maps the values $\vec{v} = \{v_i\}_{i=1}^n$ to a distribution $g_{\vec{v}}$ over the signals $\Gamma$. The function is ``independent'' for each agent: There is a set $\Gamma_i$ of signals for each agent $i$. The sender, after observing $v \sim D_i$, maps it to a distribution $g_{iv}$ over signals in $\Gamma_i$. The sender then generates each signal $\sigma_i \sim g_{iv}$ and sends the set of generated signals $\sigma = \{\sigma_i\}_{i=1}^n$ to the receiver. The receiver computes a per-agent posterior $D_i(\sigma_i)$ (where $\{D_i(\sigma_i)\}_{i = 1}^n$ are mutually independent), and hence the set $S(\sigma)$. We can alternately think that there is a separate sender for each agent that outputs a signal for that agent independent of the behaviors of other senders. This policy is known to the receiver. 
% \item After receiving the signal $\sigma \in \Gamma$, the receiver uses Bayes' rule and her knowledge of the mapping policy to compute the independent posterior distributions $\{D_i(\sigma)\}_{i = 1}^n$ over agent values. The \emph{\slr} selects one agent in the set $S(\sigma) = \argmax_i \E[D_i(\sigma)]$ either deterministically or probabilistically.
% \end{itemize}

The \emph{\mpr} is a function that maps the values $\vec{v} = \{v_i\}_{i=1}^n$ to a distribution $g_{\vec{v}}$ over the signals $\Gamma$. The function is ``independent'' for each agent: There is a set $\Gamma_i$ of signals for each agent $i$. The sender, after observing $v_i \sim D_i$, maps it to a distribution $g_{iv}$ over signals in $\Gamma_i$. The sender then generates each signal $\sigma_i \sim g_{iv}$ and sends the set of generated signals $\sigma = \{\sigma_i\}_{i=1}^n$ to the receiver. The receiver computes a per-agent posterior $D_i(\sigma_i)$ (where $\{D_i(\sigma_i)\}_{i = 1}^n$ are mutually independent). We can alternatively think that there is a separate sender for each agent that outputs a signal for that agent independent of the behaviors of other senders. This policy is known to the receiver. 

After receiving the signal $\sigma \in \Gamma$, the receiver uses Bayes' rule and its knowledge of the mapping policy to compute the independent posterior distributions $\{D_i(\sigma)\}_{i = 1}^n$ over agent values. In our model, the receiver is a utilitarian welfare maximizer. Therefore, it will select an agent with the largest posterior mean -- that is, an agent in the set $S(\sigma) := \argmax_i \E[D_i(\sigma)]$. In the case of tie-breaking (i.e., when $|S(\sigma)| > 1$), we assume that the sender can tell the receiver who to select within the set $S(\sigma)$ via a \emph{\slr{}} that picks one agent in the set $S(\sigma)$ either deterministically or probabilistically. The receiver will follow the recommendation of the \slr{}.

%\paragraph{Independent Mapping Policies}
%In this paper, motivated by the applications mentioned before, we only consider \emph{independent mapping policies}, which are a special case of mapping policies. In an \emph{independent mapping policy}, there is a set $\Gamma_i$ of signals for each agent $i$. The sender, after observing $v \sim D_i$, maps it to a distribution $g_{iv}$ over signals in $\Gamma_i$. The sender then generates each signal $\sigma_i \sim g_{iv}$ and sends the set of generated signals $\sigma = \{\sigma_i\}_{i=1}^n$ to the receiver. The receiver computes a per-agent posterior $D_i(\sigma_i)$ (where $\{D_i(\sigma_i)\}_{i = 1}^n$ are mutually independent), and hence the set $S(\sigma)$. We can alternately think that there is a separate sender for each agent that outputs a signal for that agent independent of the behaviors of other senders.

\paragraph{\SGPs{}.}
A \emph{\sgp} is a distribution $\Omega$ over \iss{}s. The receiver draws an \iss{} $\omega \sim \Omega$ from this distribution and implements it. We use $\Omega$ to denote both the \sgp{} and the distribution.\footnote{One can also define more general versions of \sgps{} that create signals with more correlation among the agents, but our definition aligns with our motivating practical applications.}

\subsection{Fairness via Approximate Majorization}  
Given a \sgp{} $\Omega$, we use $U_i(\Omega)$ to denote the expected utility that agent $i$ receives, where the expectation is over the randomness in the \sgp{} and the value distributions of the agents. This is formalized in \cref{def:expected_utility}. 

\begin{definition} [expected utility]
\label{def:expected_utility}
The \emph{expected utility} of agent $i$ from $\Omega$ is
\[
U_i(\Omega) =  \E_{\omega \sim \Omega} \left[ \sum_{\sigma \in \Gamma} q_{\omega}(\sigma)  \cdot \pw_i(\sigma) \cdot \mu_i(\sigma) \right].
\]
In this formula, after $\omega \sim \Omega$ is drawn, $q_{\omega}(\sigma)$ denotes the probability (over agent values $\vec{v} \sim \prod_i D_i$) that the receiver observes the joint signal $\sigma$; $\pw_i(\sigma)$ denotes the probability that agent $i$ is selected when the signal is $\sigma$; $\mu_i(\sigma) = \E[D_i(\sigma)]$ denotes the posterior mean value of agent $i$ given signal $\sigma$.
\end{definition}

%\yh{Some notation changes for inconsistency:
%\begin{itemize}
%    \item [1.] In this section, the original $p_i(\sigma)$ is changed to $\pw_i(\sigma)$. (\textbackslash pw)
%    \item [2.] The $p_i$ used in the Bernoulli distribution is changed to $\pb_i$ (\textbackslash pb). 
%    \
%    \item [3.] The $I$ in the Bernoulli network proof is changed to $T$ to avoid conflict with the interval buckets.
%\end{itemize}}

The goal of the sender is to design a \sgp{} $\Omega$ that is \emph{fair} with respect to the expected utilities $\{U_i(\Omega)\}_{i=1}^n$. This is captured by the notion of $\alpha$-majorization \cite{goel2005approximate,goel2006simultaneous}, an extension to the classical mathematical notion of majorization \cite{karamata1932inegalite,hardy1934inequalities}.

\begin{definition} [$\alpha$-majorization]
\label{def:alpha_majorization}
For $\alpha \ge 1$, a \sgp{} $\Omega$ is called \emph{$\alpha$-majorized}\footnote{The notion of ``$1$-majorized'' is also known as ``least weakly supermajorized'' in the literature~\cite{tamir1995least}.}  if for any $k \in \{1,2,\ldots,n\}$ and any \sgp{} $\Omega'$, the sum of the $k$ smallest utilities in $\{U_i(\Omega)\}_{i=1}^n$ is at least $1/\alpha$ times the sum of the $k$ smallest utilities in $\{U_i(\Omega')\}_{i=1}^n$.
\end{definition}

Sometimes, we use the phrase ``$\alpha$-majorization within a class $\mathcal{C}$ of signaling policies''. Its definition is to replace ``any signaling policy $\Omega'$'' to ``any signaling policy $\Omega' \in \mathcal{C}$'' in \cref{def:alpha_majorization}.

Denote $U(\Omega)$ as the vector $\{U_i(\Omega)\}_{i=1}^n$. The following result shows that approximate majorization is equivalent to simultaneously approximating all symmetric and concave welfare functions. It also holds when restricting to any class of signaling policies.

\begin{proposition} [adapted from \cite{goel2006simultaneous}]
The signaling policy $\Omega$ is \emph{$\alpha$-majorized} if and only if for every symmetric and concave function\footnote{Such a function will also be monotonically non-decreasing.} $f \colon \mathbb{R}_{\geq 0}^n \to \mathbb{R}_{\geq 0}$ (called a \emph{welfare function} or \emph{fairness function}) and any other signaling policy $\Omega'$, it holds that 
\[
f\left( U(\Omega) \right) \ge \frac{1}{\alpha} \cdot f\left( U(\Omega') \right).
\]
\end{proposition}

\begin{example}
\label{eg:maxmin}
    There are $n = 3$ agents. The value distributions $D_1, D_2$ of agents $1$ and $2$ are identical, and each of them takes value of $1$ with probability $0.5$ and takes value of $5$ with probability $0.5$. The value of agent $3$ is deterministically $2$. Consider the max-min fairness function: $f(U(\Omega)) = \min_{i = 1}^n U_i(\Omega)$, so that the sender aims at maximizing the smallest expected utility among all the agents.

    If the sender does not reveal any information about the agents, the receiver will select an agent from $\{1, 2\}$ because their expected values are both $3$, which is larger than agent $3$'s value of $2$. Since agent $3$ is never selected, the max-min welfare is $0$. 
    
    Consider the following \sgp{} $\Omega$ that implements an \iss{} $\omega$. The signal sets for each agent are:  $\Gamma_1 = \{s_1, s_1'\}$, $\Gamma_2 = \{s_2, s_2'\}$, $\Gamma_3 = \{s_3\}$. The \mpr{} of $\omega$ is given by the following probabilities:
    \begin{gather*}
        \Pr[\sigma_1 = s_1 \mid v_1 = 1] = 1;\ \Pr[\sigma_1 = s_1 \mid v_1 = 5] = \frac{1}{3};\ \Pr[\sigma_1 = s_1' \mid v_1 = 5] = \frac{2}{3};\\
       \Pr[\sigma_2 = s_2 \mid v_2 = 1] = 1;\ \Pr[\sigma_2 = s_2 \mid v_2 = 5] = \frac{1}{3};\ \Pr[\sigma_2 = s_2' \mid v_2 = 5] = \frac{2}{3};\\
        \Pr[\sigma_3 = s_3] = 1. 
    \end{gather*}   
 %   To interpret these probabilities, if agent $1$'s (resp. agent $2$'s) value is $5$, the sender sends $s_1$ (resp. $s_2$) with probability $0.25$, and sends signal $s_1'$ (resp. $s_2'$) with probability $0.75$, and so on. %If agent $1$'s (resp. agent $2$'s) value is $1$, the sender sends signal $s_1$ (resp. $s_2$) with probability $1$. 
The expected value of agent $1$ in the receiver's posterior upon receiving signal $s_1$ is 
    \begin{align*}
        \E[D_1(s_1)] & = \E[v_1 \mid \sigma_1 = s_1]\\
        & = 1 \cdot \Pr[v_1 = 1 \mid \sigma_1 = s_1] + 5 \cdot \Pr[v_1 = 5 \mid \sigma_1 = s_1] \\
        & = 1 \cdot \Pr[\sigma_1 = s_1 \mid v_1 = 1] \cdot \frac{\Pr[v_1 = 1]}{\Pr[\sigma_1 = s_1]} + 5 \cdot \Pr[\sigma_1 = s_1 \mid v_1 = 5] \cdot \frac{\Pr[v_1 = 5]}{\Pr[\sigma_1 = s_1]} \\
        & = 1\times 0.75 + 5 \times 0.25 = 2.
    \end{align*}
    Note that  $\E[D_1(s_1')] = 5$. Similarly, we have $\E[D_2(s_2)] = 2$ and $\E[D_2(s_2')] = 5$. Since $D_3$ is deterministic, we have $\E[D_3(s_3)] = 2$.
    
    The \slr{} of $\omega$ is to always break ties in favor of agent $3$ and arbitrarily break ties between agent $1$ and $2$. Since agent $3$ is selected when $s_1$ and $s_2$ are sent to the receiver, we have
    \[
    U_3(\Omega) = \left(\frac{1}{3} \times 0.5 + 1 \times 0.5\right)^2 \times 2 = \frac{8}{9}.
    \]
    Since agent $1$ is selected with value $5$ as long as $s_1'$ and $s_2$ are sent, we have 
    \[
    U_1(\Omega) \ge \left(0.5 \times \frac{2}{3}\right) \times \left(\frac{1}{3} \times 0.5 + 1 \times 0.5\right) \times 5 = \frac{10}{9} > \frac{8}{9} = U_3(\Omega).
    \]
    Similarly, we have $U_2(\Omega) > U_3(\Omega)$. Therefore, $f(U(\Omega)) = \min \{U_1(\Omega), U_2(\Omega), U_3(\Omega)\} = 8/9$. %Therefore, compared to not revealing anything to the receiver, $\Omega$ achieves a higher minimum expected utility among the three agents.
    We note that $\Omega$ is also the optimal signaling policy for the max-min welfare objective.
\end{example}

\subsection{Our Results} 
\paragraph{\fui{} Policy.} If the sender aims at maximizing the utilitarian welfare $\sum_{i=1}^n U_i(\Omega)$ (which corresponds to the welfare function $f(\vec x) = \sum_{i=1}^n x_i$), then one optimal \mpr{} is to fully reveal the values $\{v_i\}_{i=1}^n$. In this case, the sender's goal perfectly aligns with the receiver's selection behavior, obtaining the best possible utilitarian social welfare of $\E_{\vec{v} \sim \prod_i D_i} \left[\max_{i=1}^n v_i\right]$. We term this \mpr{} \fui{}. It is easy to check that when $\{D_i\}_{i = 1}^n$ are i.i.d., \fui{} paired with a symmetric \slr{} is $1$-majorized. Also, we call an \iss{} \fui{} if its \mpr{} is \fui{}.

Our first result shows that when each $D_i$ is $\ber(\pb_i)$ (which takes value $1$ with probability $\pw_i$ and value $0$ otherwise), there is a $1$-majorized \fui{} \iss{} -- it simultaneously optimizes all welfare functions, not just the utilitarian welfare. Though this result is for specific value distributions, its main idea of reducing to network flow will be useful for an approximate majorization result in \cref{sec:approx} for more general distributions. The selection policy can be computed in polynomial time. (Here and later when we mention ``polynomial time'' results, we assume that the value distributions are discrete.)

\begin{theorem}[Proved in \cref{sec:network_flow}]
\label{thm:full1}
When each $D_i$ is $\ber(\pb_i)$, \fui{} is $1$-majorized when paired with a certain polynomial-time-computable selection policy.
\end{theorem}

To prove the above theorem, we reduce the problem to network flow and apply a seminal majorization result for flows from~\cite{megiddo1974optimal} (see \cref{thm:bern_main}). A similar technique shows that even for general distributions, when restricted to the class of \fui{} policies, there is a $2$-majorized policy.

\begin{theorem}[Proved in \cref{sec:full_maj}]
\label{thm:full3}
For general distributions $\{D_i\}_{i = 1}^n$, \fui{}, when paired with a certain polynomial-time-computable selection policy, is $2$-majorized within the space of \fui{} policies.
\end{theorem}

Note that \cref{thm:full3} shows majorization \emph{within} the class of \fui{} policies. In contrast, the example below shows that if we consider the space of \emph{all} signaling schemes, \fui{} policies cannot achieve a good approximate majorization. This impossibility result motivates us to study other classes of signaling policies.

\begin{example}
\label{ex:full2}
The distribution $D_1$ is deterministically $2$, and each of the distributions $D_2, \ldots, D_n$ is $1$ with probability $1/2$ and $3$ with probability $1/2$. Consider the welfare function of max-min fairness where  $f(U(\Omega)) = \min_i U_i(\Omega)$. In \fui{}, $U_1 = 2 / 2^{n - 1}$, and therefore the max-min welfare for \fui{} is at most $2 / 2^{n - 1}$. (To calculate the expected utilities of other agents, we have $\E[\max_{i=2}^n D_i] = 3 \cdot (1-1/2^{n - 1}) + 1 \cdot (1/2^{n - 1}) \approx 3$ when $n$ is large. Therefore, for a symmetric selection policy, $U_2 = \cdots = U_n \approx 3 / n$.) In contrast, if the policy reveals nothing, all agents have posterior mean $2$. Assuming the selection policy assigns uniformly at random, $U_1 = U_2 = \cdots = U_n = 2 / n$. Therefore, \fui{} is no better than a $\left(2^{n-1} / n\right)$-approximation to the max-min fairness objective. 
\end{example}

%\begin{theorem}[Proved in \cref{sec:full2}]
%\label{thm:full2}
%When the support of each $D_i$ is $\{1,2,3\}$, regardless of the selection policy used by the receiver, \fui{} policies are not $o(2^{n}/n)$-majorized in the class of all signaling policies. Here, $n$ is the number of agents.
%\end{theorem}

%Note that any policy on the above instance is $V$-majorized, so the above lower bound rules out any non-trivial guarantee for the \fui{} policy.

\paragraph{Approximate Majorization.} We now consider the question of majorization among \emph{all} \sgps{}. For the rest of our results, we assume that each $D_i$ is supported on $[1,V]$. (This interval is equivalent to any $[v_{\min},v_{\max}]$ with $V = v_{\max} / v_{\min}$ via scaling.) Our results characterize the approximation to majorization as functions of $V$. The main hurdle with formulating a mathematical program for finding such majorized \sgps{} is the non-linear and non-convex interaction between variables encoding mapping and those encoding selection. 

%\km{I am not able to extend the approximation to be in $[0,V]$ since adding $\delta$ to the objective makes the result follow from~\cite{DBLP:journals/algorithmica/GoelM06}.}

We overcome this challenge in two steps: First, we split the utility of a fixed policy into buckets. Second, we efficiently find a mapping that ensures a large utility in some bucket by solving a simple LP with constraints on posterior means, enabling a reduction to the network flow approach from before. For the second step, we consider a specific type of \mpr{}s which we call \sgm{}. Such a \mpr{} picks a common range $[\mu, \hat{\mu}]$ for all agents and finds a mapping for each agent that maximizes the probability that the mean of the resulting signal lies in the range. We will choose $\hat{\mu} = (1+\epsilon) \mu$ for some constant $\epsilon > 0$. The \sgps{} we consider randomize over \sgm{} \mpr{}s for different $\mu$, and we will describe the corresponding \slr{}s later. Note that the \mpr{} for any agent does not depend on other agents, so that the mapping can be constructed in a distributed way.

We show that this class of policies suffices to achieve a non-trivial positive approximation result for majorization. Our result is \emph{bicriteria}, in the following sense. Let $Z_k$ denote the maximum achievable value of the sum of the smallest $k$ utilities. For given $\epsilon > 0$, we say that a policy $\Omega$ is a bicriteria $\alpha$-approximation if the following hold. 
\begin{itemize}
\item We allow $\Omega$ to use a $(1+\epsilon)$-approximate welfare maximizer\footnote{We note that polynomial-time-computability results of~\cite{dughmi2016algorithmic} also assume an approximately optimal receiver.} -- that is, for any signal $\sigma$, the \slr{} can select any agent $i$ with $\E[D_i(\sigma)] \ge \max_j \E[D_j(\sigma)] / (1+\epsilon)$.
\item For each $k$, the sum of the smallest $k$ utilities in $\Omega$ is at least $Z_k / \alpha$. Note that the quantity $Z_k$ assumes the receiver is an exact welfare maximizer.
\end{itemize}

%\kw{todo: explain what it is and point out it's standard} meaning that for constant $\epsilon > 0$, in the solution we find, the receiver is a $(1+\epsilon)$-approximate welfare maximizer -- that is, for any signal $\sigma$, the selection policy can select any agent $i$ with $\E[D_i(\sigma)] \ge \max_j \E[D_j(\sigma)] / (1+\epsilon)$. 

%For stating this result, we need the notion of {\em relaxed} approximation. 

%\begin{definition}[$(\epsilon,\beta)$-relaxed approximation]
%\label{def:relax2}
%Suppose we are optimizing some fairness objective $f$ assuming a $\theta$-welfare maximizer (see \cref{def:relax}). Assume the optimal policy is $\Omega^*$. Then a signaling policy is a $(\epsilon, \beta)$- relaxed approximation for $\beta \ge 1, \epsilon > 0$ if the solution $\Omega$ satisfies
%\begin{itemize}
%\item It assumes a $(1+\epsilon) \cdot \theta$-welfare maximizer;
%\item  $\beta \cdot \sum_{i=1}^n f(U_i(\Omega) ) \ge \sum_{i=1}^n f(U_i(\Omega^*))$.
%\end{itemize}
%We say a policy is $(\epsilon,\beta)$-majorized if for all fairness functions $f$, it is a $(\epsilon,\beta)$-relaxed approximation for the objective $f$.
%\end{definition}

\begin{theorem}[Proved in \cref{sec:approx}]
\label{thm:approx}
For any constant $\epsilon > 0$, there is a bicriteria $ O\left((\log V) / \epsilon\right)$-majorized polynomial-time-computable \sgp{}. %that assumes a $(1+\epsilon)$-approximate welfare maximizing receiver.
\end{theorem}

\begin{remark}
The generic result of \cite{goel2006simultaneous} can be applied to our setting and give an $O\left(\min\left\{n,\log \frac{P_{\max}}{n \cdot P_{\min}}\right\}\right)$-majorized solution, where $P_{\min} = \max_{\Omega} \min_i U_i(\Omega)$ and $P_{\max} = \max_{\Omega} \sum_i U_i(\Omega)$. %Since the utilities depend on the probabilities in $\{D_i\}_{i = 1}^n$,
In our setting, the former quantity could be exponentially (in $n$) smaller than the latter, leading to the approximation ratio having linear dependence on $n$. In contrast, our main result has no dependence on $n$ and only logarithmic dependence on $V$, the range of values.
\end{remark}

At a technical level, \cref{lem:struct} characterizes the structure about the class of \sgm{} mappings, and enables us to pre-compute the mapping variables. We can therefore decouple the mapping variables from the selection variables, and reduce the non-convex mathematical program to a network flow problem. Finally, we use the majorization result in \cref{thm:bern_main} to complete the proof.  

As an immediate corollary, this shows a polynomial-time $O((\log V) / \epsilon)$ bicriteria approximation algorithm for any given fairness function $f$ when agents generate mappings in a distributed fashion, which is desirable for the applications of fair selection. We note that the FPTAS in~\cite{dughmi2016algorithmic} assumes that the sender can use all agents' values to generate a common signal.

%\begin{corollary}
%\label{cor1}
%   When the sender's mapping for agents are independent,  for any fairness function $f$, there is a polynomial time computable policy that achieves a bricriteria $ O\left(\frac{\log V}{\epsilon}\right)$ approximation.
%\end{corollary}

In \cref{sec:asym}, we extend the utility model to be asymmetric between the agents and the receiver, capturing the setting in~\cite{AuKawai}. In this model, the signals and the receiver's action depend on the quality of the agents, while an agent gets utility $1$ if selected and $0$ otherwise. We show that \cref{thm:approx} extends as is, and shows approximate majorization even in this setting.

We complement the previous theorem with a lower bound showing a dependence on $V$ in the approximation factor is unavoidable, even if the receiver is an approximate welfare maximizer.

\begin{theorem}[Proved in \cref{sec:main_lb}]
\label{thm:main_lb}
No \sgp{} can be $(\log \log V) / 3$-majorized, even if the \slr{} can \emph{arbitrarily} select any agent (i.e., without the constraint that the selected agent has near-optimal expected value.). %\footnote{The lower bound holds even if the receiver is a $(1+\epsilon)$-approximate welfare maximizer.}
\end{theorem}

\section{Warmup: The \fui{} \SGP{}} 
\label{sec:full}
In this section, we consider the \fui{} \sgp{}, in which the sender reveals all the values $\vec{v} = \{v_i\}_{i=1}^n$ to the receiver. The sender can still design the \slr{} based on the fairness objective. In this case, the \slr{} specifies which agent -- among the ones that have the highest value -- the receiver should select.

%The receiver's choice of selection policy can depend on the fairness objective. Note that in this setting, the mapping is fixed and reveals the agents' values, and the flexibility is in the selection policy. 

We propose a formulation of the problem based on network flow, and en route develop techniques that will also be useful in \cref{sec:approx}. In this section, using these techniques, we show that \fui{} can be $1$-majorized when the value distributions are Bernoulli (but can be heterogeneous). We then show that for general value distributions, when restricted to the space of \fui{} \sgps{}, there is a $2$-majorized, polynomial-time-computable \sgp{}.

%We develop a technique based on a reduction to network flow. This technique will be useful for deriving the approximation result for general distributions in \cref{sec:approx}. Using this technique, we show that this policy is $1$-majorized when the value distributions are Bernoulli (but can be heterogeneous). We then show that when restricted to the space of \fui{} policies, there is a $2$-majorized, polynomial-time-computable policy.  %We however show that the space of \fui{} policies cannot achieve any non-trivial approximation to majorization among the class of {\em all} policies even when the value distributions have support size three.

\subsection{\texorpdfstring{$1$}{1}-Majorization for Bernoulli Distributions\texorpdfstring{: Proof of \cref{thm:full1}}{}}
\label{sec:network_flow}
Consider the Bernoulli setting in which each value distribution $D_i$ is $\ber(\pb_i)$ (which takes value $1$ with probability $\pb_i$ and takes value $0$ otherwise). We first show in \cref{lem:convert_to_full} that any \sgp{} can be converted to using a \fui{} \sgp{} without decreasing the utility of any agent. This reduction allows us to only care about \fui{} policies. We subsequently map each \fui{} policy to a network flow instance, thus enabling us to apply the majorization result for network flows~\cite{veinott1971least}.

To prove \cref{lem:convert_to_full}, we construct a new \fui{} signaling policy $\Omega'$ that reveals the full value vector of all the agents to the receiver. We carefully set the selection policy of $\Omega'$ so that each agent's winning probability conditioned on their value being $1$ does not decrease in $\Omega'$ compared to $\Omega$. 

\begin{lemma}
\label{lem:convert_to_full}
    When $D_i = \ber(\pb_i)$, for any \sgp{} $\Omega$, there exists a \fui{} \sgp{} $\Omega'$ in which the utility of each agent is at least their utility in $\Omega$.
\end{lemma}
\begin{proof}
    Given any \sgp{} $\Omega$, let $\Gamma$ be the set of all possible signals that can be sent to the receiver in $\Omega$. We will construct a \fui{} signaling policy $\Omega'$ that preserves the utilities of all agents. 
    
    For each $\sigma \in \Gamma$, recall that $D_i(\sigma)$ is the receiver's posterior Bernoulli distribution over agent $i$'s value. %\kw{I didn't understand what the previous sentence is trying to say.} 
    Denote by $q_i(\sigma)$ the conditional probability that agent $i$ is selected when $\sigma$ is sent to the receiver. 
    
    In $\Omega$, suppose the sender observes the value vector of agents as $\vec{v}_0$ and sends signal $\sigma \in \Gamma$. Then, in $\Omega'$, the sender reveals $\vec{v}_0$ to the receiver. Suppose that $S(\vec{v}_0)$ is the set of agents with value $1$ in $\vec{v}_0$. In $\Omega'$, the selection policy of the receiver chooses $i \in S(\vec{v}_0)$ with probability $q_i(\sigma)$. If $\sum_{i \in S(\vec{v}_0)} q_i(\sigma) < 1$, the selection probabilities for $S(\vec{v}_0)$ are arbitrarily increased so that they sum to $1$. Denote $v^i_0$ as the $i^{th}$ component of $\vec{v_0}$.  Noting that an agent only gains utility when their value is $1$, in the new \sgp{} $\Omega'$, the expected utility of agent $i$ is 
    \begin{align*}
    U_i(\Omega') &\ge \sum_{\vec{v}_0 :\ v^i_0 = 1}  \Pr[\vec{v} = \vec{v}_0] \cdot \left(\sum_{\sigma \in \Gamma} \Pr[\text{$\sigma$ is sent in $\Omega$} \mid \vec{v} = \vec{v}_0] \cdot q_i(\sigma)  \right)\\
    &= \sum_{\sigma \in \Gamma} \left(\sum_{\vec{v}_0 :\ v^i_0 = 1} \Pr[\text{$\sigma$ is sent in $\Omega$} , \ \vec{v} = \vec{v}_0] \cdot q_i(\sigma) \right)\\
    &= \sum_{\sigma \in \Gamma} \left(\sum_{\vec{v}_0 :\ v^i_0 = 1} \Pr[\vec{v} = \vec{v}_0 \mid \textrm{$\sigma$ is sent in $\Omega$}] \right) \cdot \Pr[\text{$\sigma$ is sent in $\Omega$}] \cdot q_i(\sigma) \\
    &= \sum_{\sigma \in \Gamma} \Pr[v_i = 1 \mid \textrm{$\sigma$ is sent in $\Omega$}] \cdot \Pr[\textrm{$\sigma$ is sent in $\Omega$}] \cdot q_i(\sigma).    
    \end{align*}
     In $\Omega$, conditioned on $\sigma$ being sent, the expected utility of agent $i$ is $q_i (\sigma) \cdot \E[D_i(\sigma)]$. Since $D_i(\sigma)$ is Bernoulli, we have $\E[D_i(\sigma)] = \Pr[v_i = 1 \mid \textrm{$\sigma$ is sent in $\Omega$}]$. Therefore,
    \[
    U_i(\Omega) = \sum_{\sigma \in \Gamma}  \Pr[v_i = 1 \mid \textrm{$\sigma$ is sent in $\Omega$}] \cdot \Pr[\textrm{$\sigma$ is sent in $\Omega$}] \cdot q_i(\sigma) \le U_i(\Omega').
    \]
    This completes the proof.
\end{proof}

\paragraph{Reduction to Network Flow.} We now construct the network flow instance from the \fui{} policies. We first build a bipartite directed flow graph $G$ as follows:

\begin{itemize}
    \item There is a single source $s$. Place a set $R$ of $n$ sinks $t_{1}, \ldots, t_{n}$ where $t_{i}$ represents the agent $i$. %\yh{\sout{Add an edge $(t_{i}, t)$ with capacity $1$.}}
    \item Place a set $L$ of $2^n$ nodes where each subset $T \subseteq [n]$ corresponds to a node $s_T$. 
    \item Add an edge $(s, s_T)$ with capacity $\pb_T = \prod_{i \in T} \pb_i \cdot \prod_{i \notin T} (1 - \pb_i)$. For each agent $i \in T$, add an edge $(s_T, t_i)$ with capacity $+\infty$.  
\end{itemize} 

We illustrate this construction on an example with $n = 2$ agents in \Cref{fig:network_1}. Their value distributions are $D_1 = \ber(0.3)$ and $D_2 = \ber(0.6)$. Since there are two agents, we have $4$ nodes in $L$, denoted by $s_{\emptyset}$, $s_{\{1\}}$, $s_{\{2\}}$ and $s_{\{1,2\}}$. By the construction above, we have the capacities on edges from $s$ to these four nodes are $(1 - 0.3) \times (1 - 0.6) = 0.28$, $0.3 \times (1 - 0.6) = 0.12$, $(1 - 0.3) \times 0.6 = 0.42$ and $0.3 \cdot 0.6 = 0.18$, respectively. 

%\yh{Add a picture to show the network flow reduction.}
%\begin{example} \label{ex:network}
%    There are $2$ agents with $D_1 = \ber(0.3)$ and $D_2 = \ber(0.6)$. The corresponding constructed network is shown in \cref{fig:network_1}.
%\end{example}

\begin{figure}[htbp]
    \centering
    \includegraphics[scale = 1.0] {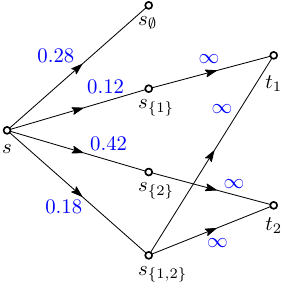}
    \caption{The network flow instance for an example with $n = 2$ agents, with $D_1 = \ber(0.3)$ and $D_2 = \ber(0.6)$. The number in blue next to each directed edge indicates the capacity of the edge. }
    \label{fig:network_1}
\end{figure}

The intuition for the construction and  \cref{lem:network_bernoulli} is the following. Each node in $L$ corresponds to a possible value vector of all agents, or equivalently, the set of agents whose values are $1$. The flow value from the source to such node is capped by the probability that this value vector is realized. The flow value from a node in $L$ to an agent node equals the probability of the event that the vector corresponding to $L$ is realized and this agent is selected.

\begin{lemma} \label{lem:network_bernoulli}
The set of feasible $s$--$t$ maximum flows in $G$ corresponds exactly to the set of \fui{} \sgps{}. For any policy $\Omega$, the utility $U_i(\Omega)$ of agent $i$ is equal to the total flow $f_i$ to the sink $t_i$ in $G$ in the corresponding flow instance.
\end{lemma}
\begin{proof}
We first prove that any \sgp{} $\Omega$ can be converted to a maximum flow. Fix any set $T \subseteq [n]$, and consider the event in which the agents with value $1$ exactly form the set $T$. Let $z_{Ti}$ be the probability that conditioned on this event, the agent $i \in T$ is selected in $\Omega$. We set flow $f_{Ti}$ on the edge $(s_T, t_i)$ to be $\pb_T \cdot z_{Ti}$. Note that $U_i(\Omega) = \sum_{T : i \in T} \pb_T \cdot z_{Ti} = \sum_{T : i \in T} f_{Ti} = f_i$. The total flow is $\sum_{T \neq \varnothing} \pb_T$, which means that it is a maximum flow.

For the other direction, given a maximum flow, let $T$ denote the set of agents with value $1$. Notice that when $T \neq \varnothing$, the flow from $s$ to $s_T$ must be equal to its capacity. We construct a \fui{} \sgp{} $\Omega$ by requiring that conditioned on $T$ being the set of agents with value $1$, the receiver selects $i \in T$ with probability $f_{Ti} / \pb_T$, where $f_{Ti}$ is the flow on the edge $(s_T, t_i)$. This is clearly a feasible signaling policy with $U_i(\Omega) = \sum_{T : i \in T} \pb_T \cdot (f_{Ti} / \pb_T) = f_{i}$.
\end{proof}

The following theorem (implicit in~\cite{veinott1971least}) shows the existence of a $1$-majorized flow, and therefore shows the existence of a $1$-majorized \fui{} \sgp{}. We include a short proof for completeness.

\begin{theorem}[$1$-Majorized Flows, Implicit in \cite{veinott1971least}]
\label{thm:bern_main}
Given any capacitated network with a single source and multiple sinks $t_1, \ldots, t_n$, let $f = \{f_i\}_{i=1}^n$ be the flows entering the respective sinks in a feasible flow. Then there is a flow $f^*$  that is $1$-majorized among these feasible flows.
\end{theorem}
\begin{proof}
    By \cite[Lemma 4.1]{megiddo1974optimal}, the set of feasible $\{f_i\}_{i=1}^n$ defines a polymatroid -- that is, for all $T \subseteq [n]$, it holds that $\sum_{i \in T} f_i \le g(T)$ and $\{f_i\}_{i=1}^n \ge 0$, where $g$ is a non-negative, monotone, submodular function. By \cite[Theorem 3.2]{tamir1995least} (see also~\cite{dutta1989concept}), any polymatroid has a $1$-majorized point.
\end{proof}

\paragraph{Polynomial-Time Algorithm.}
Though the graph $G$ constructed above has exponential size, it follows from~\cite{megiddo1974optimal} that a flow $\{f_i\}_{i=1}^n$ is feasible if and only if it is feasible for the following polymatroid $\P$:
\begin{alignat*}{2}
    \sum_{i \in T} f_i &\le \sum_{J : J \cap T \neq \varnothing} \pb_{J} = 1 - \prod_{i \in T} (1 - \pb_i), & \qquad & \forall T \subseteq [n];\\
    f_{i} &\in [0,1], & \qquad & \forall i \in [n].
\end{alignat*}

Maximizing any strictly concave, symmetric, separable function over this polymatroid now yields the solution in polynomial time~\cite{tamir1995least,veinott1971least}. (See~\cite{goel2006simultaneous} for an LP based approach.) This can be solved to an arbitrary approximation in polynomial time. %Using parametric search, we find the smallest $x$ such that for some $i$, $f_i < x$. This yields $x_i$ and $a_i$. \kw{What does this mean? Also let me rewrite this part...}
%Now, the main result in~\cite{DBLP:journals/algorithmica/GoelM06} shows how to find a majorized solution to such linear programs via iteratively adding $O(n)$ constraints and solving it. Since polymatroid constraints admit to efficient separation oracles, this yields a polynomial time algorithm to find a $1$-majorized solution via either the Ellipsoid algorithm or the multiplicative weight method.

%Next, the proof of \cref{thm:bern_main} shows the following: There is an ordering of agents, call it $a_1, a_2, \ldots, a_n$ and corresponding time points $x_1 \le x_2 \le \cdots x_n$ so that for any $a_i$ and for all $x \ge x_i$, in graph $G(x)$, the monotone maximum flow $f$ has $f_{a_i} = x_i$. In other words, the flow goes ``tight'' in this order.

%We compute $x_i$ sequentially. Suppose we have identified $a_1,\ldots,a_{i-1}$. Let $f^0$ denote the solution for $G(x_{i-1})$ Then for $x \ge x_{i-1}$, we write the optimization problem:

%\[ \mbox{Maximize } \sum_{j=1}^n \min(x, f_{a_j})\]
%over the set
%\begin{eqnarray*}
%    \forall T \subseteq [n], & \sum_{j \in T} f_{a_j} \le \sum_{T' | T' \cap T \neq \emptyset} \pb_{T'} = 1 - \prod_{j \in T} (1 - \pb_j);\\
%    \forall j \in [i-1], & f_{a_j} = x_j; \\
%    \forall j \ge i & f_{a_j} \ge f^0_{a_j};
%\end{eqnarray*}

The final solution $f^*$ for $G$ lies in the convex hull of the vertices of $\P$. By Carath\'eodory's Theorem, this point can be written as the convex combination of at most $n + 1$ vertices on the convex hull. By applying an ellipsoid method, we can find such a decomposition in polynomial time. Denote the decomposed vertices and their weights by $\{(w_k, \eta_k)\}_{k = 1}^{n + 1}$, where $\sum_{k = 1}^{n + 1} \eta_k = 1$.

By the property of polymatroids, each vertex $u_t$ corresponds to a series of ``tight subsets'' $\varnothing = T_{0} \subsetneq T_{1} \subsetneq T_{2} \subsetneq \cdots \subsetneq T_{n} = [n]$, where each adjacent pair of sets $(T_{k-1}, T_{k})$ only differ by a single agent $i_{k}$. By a ``tight subset'' $T$ we mean that the inequality $\sum_{i \in T} f_i \le 1 - \prod_{i \in T} (1 - \pb_i)$ is tight at this vertex. Since $f_i$ is the utility of agent $i$ for this vertex solution, as long as at least one agent in $T_k$ has value $1$, some agent in this set must be finally selected. This yields the following ranking scheme for $w_t$:
\begin{itemize}
    \item [1.] The sender reveals all indices $\{i_{k}\}_{k = 1}^n$.
    \item [2.] The receiver selects the first agent $i_{k}$ in the ranking whose value is $1$.
\end{itemize}

The  signaling scheme corresponding to $f^*$ is a randomization over these ranking schemes: With probability $\eta_k$, run the ranking mechanism corresponding to the vertex $w_k$. Since the utilities of the agents are exactly $\{f^*_i\}_{i=1}^n$, this is $1$-majorized. This completes the proof of \cref{thm:full1}.

\subsection{Majorization among \texorpdfstring{\textsc{Full Information}}{Full Information} Schemes\texorpdfstring{: Proof of \cref{thm:full3}}{}}
\label{sec:full_maj}
We now show that if we restrict to the space of \textsc{Full Information} mapping policies, there is a selection policy that is $2$-majorized and can be computed in polynomial time. Our technique again reduces to network flow, albeit for a relaxation of the problem. We will use a similar relaxation in \cref{sec:approx} for the general problem.

%\yh{Kangning and I discusses about this -- this method seems to be able to combine with the existence of 1-majorized scheme into a 4-majorized scheme in polynomial time (?) We can apply the LP with objective $\sum_{k = 1}^n P_k(\bu)$, where $P_k(\bu)$ is the sum of the smallest $k$ agents' utilities. As long as the 1-majorized scheme exists, this objective solves the 1-majorized utility vector.}

%Suppose that the welfare function is a con function $f(\bu)$, where $\bu = (u_1, u_2, \cdots, u_n)$ representing each agent's utility. In agent $i$'s value distribution $D_i$, assume the probability of her value being $v_j$ is $d_{ij}$.

We assume that there are $m$ possible values in total that support the value distributions $\{D_i\}_{i = 1}^n$. In the remainder of this section, we use the notation $v_1 < v_2 < \cdots < v_m$ to denote these $m$ possible values. Let $d_{ij} = \Pr[D_i = v_j]$. Let $Z_j$ be the event that the largest value is at most $v_j$. We further define
\[
z_j := \Pr[Z_j] = \prod_{i = 1}^n \left(\sum_{j' \leq j}{d_{ij'}}\right)
\]
and
\[
p_{ij} := \Pr[D_i = v_j \mid Z_j] = \frac{d_{ij}}{\sum_{j' \leq  j} d_{ij'}}.
\]

%For each agent $i$, denote her original value random variable by $x_i$, and the random Bernoulli variables by $\{y_{ij}\}_{j\in [m]}$. The original distribution is equivalent to the maximum of these $m$ decomposed distributions:

%$$
%x_i = \max_{j \in [m]} y_{ij} \text{ where } y_{ij} %= \begin{cases}
%    0 & \text{w.p. $1 - p_{ij}$;}\\ 
%    v_j & \text{w.p. $p_{ij}$}.
%\end{cases}
%$$

\paragraph{Network Flow Instance.}
We will show a network flow instance so that any \fui{} policy $\Omega$ can be relaxed into a feasible flow for this instance. Given $\Omega$, let $x_{ij}$ be the probability that agent $i$ has value $v_j$ and is selected conditioned on the event $Z_j$. %Note that $\Pr[D_i = v_j \mid Z_j] = p_{ij}$. Let $u_i = U_i(\Omega)$. %Let $y_{ij} = z_j \cdot v_j  \cdot x_{ij}$.
For the policy $\Omega$, the variables $\{x_{ij}\}$ and expected utilities $\{u_i\}_{i = 1}^n$ satisfy the following constraints.
%\begin{gather}
 %   \text{Maximize } \lambda = f(u_1, u_2, \ldots, u_n), \text{ subject to} \notag\\
%    u_i = \sum_{j=1}^m z_j \cdot v_j  \cdot x_{ij} ,\ \forall i \in [n]; \notag \\
%    \sum_i {x_{ij}} \le 1,\ \forall j; \label{eqn:constraint_2}\\
 %   1 - \prod_{i \in [n]} (1 - p_{ij}),\ \forall j; \label{eqn:constraint_2}\\
 %   0 \le x_{ij} \le p_{ij},\ \forall i, j. \label{eqn:constraint_3}
%\end{gather}

\begin{alignat}{2}
    \sum_{j=1}^m z_j \cdot v_j  \cdot x_{ij} &= u_i,  & \qquad & \forall i \in [n]; \label{eqn:constraint_1} \\
    \sum_{i = 1}^n {x_{ij}} &\le 1, & \qquad & \forall j \in [m]; \label{eqn:constraint_2}\\
    0 \le x_{ij} &\le p_{ij}, & \qquad & \forall i \in [n], j \in [m]. \label{eqn:constraint_3}
\end{alignat}
Conditioned on the event $Z_j$, \cref{eqn:constraint_2} says at most one agent with value $v_j$ is selected, and \cref{eqn:constraint_3} says the probability that agent $i$ has value $v_j$ and is selected is at most $p_{ij}$, the probability that its value is $v_j$.

Let $y_{ij} = z_j \cdot v_j  \cdot x_{ij}$, and $b_j = z_j \cdot v_j$. We can rewrite the above program as:
%\begin{gather}
%    u_i = \sum_{j=1}^m y_{ij} ,\ \forall i \in [n]; \notag \\
%    \sum_i {y_{ij}} \le b_j ,\ \forall j; \\ %\label{eqn:constraint_2}\\
%    0 \le y_{ij} \le p_{ij} \cdot b_j,\ \forall i, j.  %\label{eqn:constraint_3}
%\end{gather}
\begin{alignat*}{2}
\sum_{j=1}^m y_{ij} &= u_i, & \qquad & \forall i \in [n]; \\
\sum_{i = 1}^n {y_{ij}} &\le b_j, & \qquad & \forall j \in [m]; \\ 
0 \le y_{ij} &\le p_{ij} \cdot b_j, & \qquad & \forall i \in [n], j \in [m].  
\end{alignat*}

Construct a network flow graph $G$ with a source $s$, and one sink $t_i$ for each agent $i$. There are $m$ nodes, one for each $v_j$. There is a directed edge $(s,v_j)$ with capacity $b_j$, and a directed edge $(v_j, t_i)$ with capacity $p_{ij} \cdot b_j$. The flow value on the edge $(v_j, t_i)$ is $y_{ij}$, and sink $t_i$ receives flow $u_i$. This shows the network flow instance such that any $\Omega$ relaxes to a feasible flow for this instance, and the structure of the network is illustrated in \cref{fig:network_2}.

\begin{figure}[htbp]
    \centering
    \includegraphics[scale = 1.0] {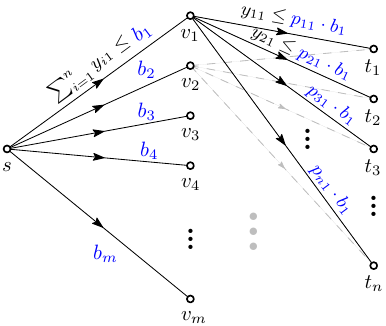}
    \caption{Illustration of the network flow instance. The text in blue denotes the capacity of each edge, while the text in black denotes the feasible flow value on each edge.  }
    \label{fig:network_2} 
\end{figure}

Using \cref{thm:bern_main}, there is a $1$-majorized flow in $G$, which translates to a $1$-majorized utility vector $\vec{u}$. Further, this solution can be computed in polynomial time~\cite{tamir1995least,veinott1971least,goel2006simultaneous}. Denote the majorized solution (in the original variables) as $\{x_{ij}\}$ and utilities as $\{u_i\}$. 

\paragraph{Signaling Policy.} We now convert the $1$-majorized solution to the above flow instance into a simple $2$-majorized \sgp{}. This is inspired by the algorithm for stochastic knapsack~\cite{DBLP:journals/mor/DeanGV08}. 

\begin{itemize}
    \item The sender sends all realized values to the receiver.
    \item Suppose that the largest realized value is $v_j$. Let $S$ denote the set of agents $i$ with value $v_j$.
    \item Sort the agents in $S$ in a uniformly random order, and consider the agents in this order.
    \item In this order, when we reach each agent $i$, with probability $x_{ij} / p_{ij}$, select $i$ and stop.
    \item If no agent in $S$ is selected above, then select an arbitrary agent in $S$.
\end{itemize}

\paragraph{Analysis.} Consider any agent $i$. Conditioned on the event $Z_j$ (i.e., the largest value being at most $v_j$), the probability that another agent $i' \in S$ is selected before agent $i$ is reached is at most
\[
\frac{1}{2} \cdot \frac{x_{i'j}}{p_{i'j}} \cdot p_{i'j} = \frac{x_{i'j}}{2},
\]
where the factor of $1/2$ comes from the fact that the order is uniformly random. Since $\sum_{i'} x_{i'j} / 2 \le 1 / 2$, the probability that no agent before $i$ is selected is at least $1 / 2$. Now, for agent $i$ with value $v_j$ to be selected, the following should happen: 
\begin{itemize} 
\item The event $Z_j$, which happens with probability $z_j$; 
\item Agent $i$ has value $v_j$ conditioned on $Z_j$, which happens with probability $p_{ij}$; 
\item No agent before $i$ is selected conditioned on all previous events, which happens with probability at least $1/2$; and
\item Agent $i$ with value $v_j$ is selected conditioned on all previous events, which happens with probability $x_{ij}/p_{ij}$.
\end{itemize}
Therefore, the unconditional probability that agent $i$ with value $v_j$ is selected is 
\[
\Pr[\text{agent } i \text{ with value } v_j \text{ is selected}] \ge z_j \cdot \frac{1}{2} \cdot p_{ij} \cdot \frac{x_{ij}}{p_{ij}} = z_j  \cdot \frac{x_{ij}}{2}.
\]
Therefore, the expected utility of agent $i$ in this scheme is at least: 

\begin{equation}
\label{eqn:utility1}
\sum_{j = 1}^m \Pr[\text{agent } i \text{ with value } v_j \text{ is selected}] \cdot v_j \ge \sum_{j = 1}^m z_j \cdot \frac{x_{ij}}{2} \cdot v_j = \frac{u_i}{2}.
\end{equation}

This shows a $2$-majorized signaling scheme running in polynomial time, proving \cref{thm:full3}.

\section{Approximate Majorization\texorpdfstring{: Proof of \cref{thm:approx}}{}}
\label{sec:approx}
We now present our main technical result -- approximate majorization for general \sgps{}. The key hurdle with using the approach in \cref{sec:network_flow,sec:full_maj} is that unlike the \fui{} setting where the \mpr{} is fixed, any mathematical program for general policies needs to encode both mapping and selection variables, and their interaction is non-linear. To deal with the challenge, we will perform a \emph{projection} of any \sgp{} onto a specific class of \sgm{} policies, for which the network flow majorization approach in the previous sections can be employed.  %We assume the receiver is $\theta$-welfare maximizing (\cref{def:relax}) for $\theta > 1$. 

%We start by assuming a given fairness function $f$. Suppose we wish to find the signaling policy $\Omega$ that maximizes $f(U(\Omega))$. %Though the work of~\cite{DughmiX16} implies a FPTAS for this problem, 
 %Let the optimal independent signaling policy for $f$ be $\Omega^*$, and let $OPT =  f(U(\Omega^*))$. We construct an approximately optimal policy for which we can show approximate majorization via reduction to network flow analogous to \cref{sec:network_flow}.

\subsection{\sgm{} Projections}
\label{sec:reduction}
Given a small constant $\epsilon > 0$, let $\eta = 1+\epsilon$. Without loss of generality, assume $V$ is a power of $\eta$, and divide the interval $[1, V]$ into buckets $I_1 = [1,\eta), I_2 = [\eta, \eta^2 ), \ldots, I_K = [V/\eta, V]$, where $K \approx (\log V) / \epsilon$. For any policy $\Omega^*$, let $u_i$ denote agent $i$'s utility, that is, $u_i = U_i(\Omega^*)$. We can express $u_i$ as the sum of contributions from different buckets as $u_i = \sum_{k=1}^K c_{ik}$, where $c_{ik}$ is the contribution to the utility $u_i$ from those signals $\sigma$ for which the posterior mean satisfies $\mu_i(\sigma) := \E[D_i(\sigma_i)] \in I_k$. 

%We immediately have the following lemma:

%\km{I feel this lemma should be improvable to a constant via prophet inequality.}

%\begin{lemma}
%\label{lem:k-approx}
%$$ \frac{OPT}{K} \le  f \left( \left \{ \frac{1}{K} \cdot  \sum_{k=1}^K c_{ik}  \right\}_{i=1}^n \right).$$
%\end{lemma}
%\begin{proof}
%Note that $OPT  = f \left( \left \{  \sum_{k=1}^K c_{ik}  \right\}_{i=1}^n \right)$. The inequality follows by the fact that $f$ is monotonically non-decreasing and concave. 
%\end{proof}

%We ignore the additive $\delta$ term from now on, since this gets absorbed into the definition of relaxed approximation (see \cref{def:relax2}). 
We now define \sgm{} projections that assume that the only contribution to utility for agent $i$ happens when the signal $\sigma$ received has posterior mean $\mu_i(\sigma) \in I_k$ and agent $i$ is selected.

%The quantity $\sum_{i=1}^n f(c_{ik})$ pretends that  We formalize this as a \sgm{} policy defined below.

\begin{definition} [\sgm{} projection]
Given an interval $I_k$ and \sgp{} $\Omega$, let $A_{ik}$ denote the event where both (I) agent $i$ is selected and (II) $\E[D_i(\sigma)] \in I_k$.  %Let  $ c_{ik}(\Omega) = \Pr[H_{ik}] \cdot \E[\mu_i(\sigma) | H_{ik}].$
The \sgm{} projection $\Omega_k$ of $\Omega$ onto $I_k$ only counts utilities from the event $\bigcup_{i = 1}^n A_{ik}$, and sets the remaining utilities to $0$. In other words, if $c_{ik}(\Omega) = \Pr[A_{ik}] \cdot \E[\mu_i(\sigma) \mid A_{ik}]$, then $U_i(\Omega_k) = c_{ik}$ for all agents $i$.  %for interval $I_k$ chooses $\Omega$ to maximize $ \sum_{i=1}^n f(c_{ik}(\Omega)).$
\end{definition}

Consider now the policy that picks a number $k \in \{1,2, \ldots, K\}$ uniformly at random and executes the \sgm{} projection $\Omega_k$ of $\Omega^*$ for this interval. The expected utility of agent $i$ is exactly $\sum_{k=1}^K c_{ik} / K = u_i / K$. We overload the notation and denote the new utilities as $\{u_i\}_{i = 1}^n$.

\subsection{Structure of \sgm{} Projections}
\label{sec:structure}
We now show that \sgm{} projections can be modified to have a surprisingly simple structure while not decreasing the utilities $\{c_{ik}\}$. Let the interval under consideration be $I_k = [m, \eta \cdot m)$ (or $[m, \eta \cdot m]$ if $k = K$) and let the \sgm{} projection of $\Omega^*$ be $\Omega_k$. Losing an additional factor $\eta$ in the approximation factor, we can assume that if $\mu_i(\sigma) \in I_k$, then the expected utility of agent $i$ is always $m$ if $i$ is selected. 

Note that if $\mu_i(\sigma) \in I_k$ and agent $i$ is selected, then in this scenario, $\max_j \mu_j(\sigma) \in [m, \eta \cdot m]$. Denote the interval $[m, \eta \cdot m]$ as $\hat{I}$, and $\eta \cdot m$ as $\hat{m}$. Assume now that when all agents have posterior mean at most $\hat{m}$, any agent whose posterior mean is in $\hat{I}$ can be selected, and this choice yields utility $m$ to that agent. Such a policy assumes an $\eta$-approximate utilitarian-welfare-maximizing receiver, and hence our final result will be bicriteria.

Clearly, whenever $\Omega_k$ selects an agent with posterior mean in $I_k$, all agents must have posterior means at most $\hat{m}$, so that the new policy can also be made to make the same selection. If we pretend the utility yielded to the agent is $m$, this is still within factor $\eta$ of the actual utility -- the modified policy has utilities $\hat{c}_{ik} \geq c_{ik} / \eta$.

%Therefore, we can consider a relaxed \sgm{} projection, where (I) the interval under consideration is $\hat{I}$, and (II) whenever all agents have posterior mean at most $\hat{m}$, any agent in $\hat{I}$ is selected, yielding utility $m$ to that agent. (The actual utility of the agent is only larger.) The modified policy has utilities $\hat{c}_{ik} \geq c_{ik} / \eta$, where the factor of $1/\eta$ comes from our assumption that the utility provided is always $m$. %Since $f$ is concave, this loses a factor of at most $\eta$ in the objective.

\paragraph{Mathematical Program.} Focus on some $k$, and denote the relaxed \sgm{} projection of $\Omega^*$ as $\Omega_k^*$. This policy can be written as a distribution over \iss{}s as 
\[
\Omega_k^* = \sum_{\ell \ge 1} \rho_{\ell} \cdot \tau_{\ell},
\]
where each \iss{} $\tau_{\ell}$ is an independent mapping of the agents to signals, along with an associated \slr{}, and $\sum_{\ell \ge 1} \rho_{\ell} = 1$. For any \iss{} $\tau$ in the above summation, let $u_i(\tau)$ denote the utility of agent $i$ in this scheme, so that
\[
U_i(\Omega_k^*) = \sum_{\ell} \rho_{\ell} \cdot u_i(\tau_{\ell}).
\]
Note that given the \sgm{} structure, we only count the utility when the posterior mean of the selected agent lies in $\hat{I}$, in which case the contribution is assumed to be $m$.

Given signaling scheme $\tau$, let $B(\tau)$ denote the event that every agent $j$ has $\mu_j(\sigma) \le \hat{m}$, where $\sigma \sim \tau$ is an independent signal. Now define the following four probabilities for the signaling scheme $\tau$:
\begin{alignat*}{3}
Q_i(\tau) & := \Pr_{\sigma \sim \tau} [ \mu_i(\sigma) \le \hat{m} ];\\
Q(\tau) & := \Pr[B(\tau)] = \prod_{i=1}^n Q_i(\tau);\\
p_i(\tau) & := \Pr_{\sigma \sim \tau} [ \mu_i(\sigma) \in \hat{I} \mid B(\tau) ];\\
x_i(\tau) & := \Pr_{\sigma \sim \tau}[\text{agent } i \text{ is selected and }\mu_i(\sigma) \in \hat{I} \mid B(\tau) ].
\end{alignat*}
$Q(\tau) = \prod_{i=1}^n Q_i(\tau)$ comes from the assumption that in the \iss{} $\tau$, the \mpr{} of each agent is independent of all other agents. The utilities in policy $\tau$ clearly satisfy: 
\begin{equation}
\label{eq:obj}
    u_i(\tau) = m \cdot Q(\tau) \cdot  x_i(\tau), \qquad \forall i \in [n].
\end{equation}
%To see this, note that $Q(\tau) \cdot x_i(\tau)$ is the probability of event $B(\tau)$, and conditioned on this, agent $i$ having posterior mean in $\hat{I}$ and being selected. 
The event $\mu_i(\sigma) \in \hat{I}$ conditioned on $B(\tau)$ follows Bernoulli$(1,p_i(\tau))$, so that  analogous to \cref{sec:network_flow}, the following polymatroid exactly captures the $\{x_i(\tau)\}_{i=1}^n$ in any feasible policy.
%\begin{equation}
%\label{eq:nonlinear}
%     \sum_{i = 1}^n x_i(\tau) \leq 1;  \qquad \text{and} \qquad  0 \leq x_i(\tau) \leq p_i(\tau), \quad \forall i \in [n].
%\end{equation}

\begin{equation}
\label{eq:nonlinear}
     \sum_{i \in S} x_i(\tau) \leq 1 - \prod_{i \in S} (1- p_i(\tau)),  \qquad \forall S \subseteq [n].
\end{equation}

%To see the first constraint in \cref{eq:nonlinear}, conditioned on the event $B(\tau)$, the left-hand side is the probability that some agent with posterior mean in $\hat{I}$ is selected, which should be at most one, the right-hand side.% The second constraint is straightforward, and this shows \cref{eq:nonlinear} holds.
%\km{Note that \cref{eq:nonlinear} is a relaxation of the set of feasible policies; however, that does not affect our proofs.}

\paragraph{Structural Lemma.} The issue now is that the quantities $Q(\tau)$ and $p_i(\tau)$ depend on the policy $\tau$, so that we cannot yet reduce it to network flow in a way analogous to \cref{sec:network_flow}. Our key insight is that this dependence can be removed. Towards this end, we define a {\em maximal mapping} as follows.

\begin{definition}[Maximal Mapping] 
   For an interval $\hat{I}$, a \emph{maximal mapping} is a \mpr{} $\omega$ from agent values to signals $\{\sigma\}$ such that  $\Pr_{\sigma \sim \omega} [\mu_i(\sigma) \in \hat{I}]$ is maximized for each agent $i$.
\end{definition}

We first show the following structural lemma, and subsequently show how to compute the maximal mapping efficiently. 

\begin{lemma}
\label{lem:struct}
Given any \iss{} $\tau$ above, there is an \iss{} $\tau'$ that (I) uses a maximal mapping, (II) satisfies \cref{eq:obj,eq:nonlinear}, and (III) gives each agent $i$ a utility of at least $u_i(\tau)$ from the signals in which their posterior mean is inside $\hat{I}$, where the contribution to the utility is assumed to be $m$.
%\yh{(Do we need to emphasize that we are still pretending that the utility from winning inside $\hat{I}$ is $m$?)}
\end{lemma}
\begin{proof}
Consider $\tau$ and the corresponding \mpr{} for agent $i$. Let $S^1_{i}$ be the set of signals $\sigma$ with the posterior mean $\mu_i(\sigma) < m$; $S^2_i$ be those with $\mu_i(\sigma) \in \hat{I}$; and $S^3_i$ those with $\mu_i(\sigma) > \hat{m}$. Create three signals $\sigma^1_i, \sigma^2_i, \sigma^3_i$. These are sent  whenever a signal in $S^1_i, S^2_i, S^3_i$ is sent, respectively. Note that $\mu_i(\sigma^1_i) = m^1_i < m$; $\mu_i(\sigma^2_i) = m^2_i \in \hat{I}$; and $\mu_i(\sigma^3_i) = m^3_i > \hat{m}$. Further,  $Q_i = \sum_{\sigma \in S^1_i \cup S^2_i} \Pr[\sigma] = \Pr[\sigma^2_i] + \Pr[\sigma^1_i]$ remains unchanged and similarly $p_i$ remains unchanged. This preserves both the constraints and the utilities.

For agent $i$, let $\alpha_i$ satisfy $\alpha_i m^1_i + (1-\alpha_i) m^3_i = m^2_i$, and define $s^1_i := \Pr[\sigma^1_i]$ and $s^3_i := \Pr[\sigma^3_i]$. Create a new signal $\sigma^4_i$, and do the following things:
\begin{itemize}
    \item If $\beta = \frac{s^1_i (1-\alpha_i)}{s^3_i \alpha_i} \le 1$, then whenever signal $\sigma^1_i$ was sent,  $\sigma^4_i$ is sent instead, and whenever $\sigma^3_i$ was sent, $\sigma^4_i$ is sent instead with probability $\beta$ and $\sigma^3_i$ is sent with probability $1-\beta$. 
    \item If $\beta > 1$, then whenever signal $\sigma^3_i$ was sent, $\sigma^4_i$ is sent instead, and whenever $\sigma^1_i$ was sent, $\sigma^4_i$ is sent instead with probability $1/\beta$ and $\sigma^1_i$ sent with probability $1-1/\beta$.
\end{itemize}
In either case, we note that $\mu_i(\sigma^4_i) = \alpha_i m^1_i + (1-\alpha_i) m^3_i = m^2_i \in \hat{I}$. The sender can then send signal $\sigma^2_i$ whenever $\sigma^4_i$ is sent, noting that this preserves the condition that $\E[\mu_i(\sigma^2_i)] \in \hat{I}$. 

Note that in either case, $\Pr[\sigma^3_i]$ does not increase, and hence $Q_i = \Pr[\sigma^1_i] + \Pr[\sigma^2_i]$ does not decrease. Since the increase in $\Pr[\sigma^2_i]$ is exactly the decrease in the sum $\Pr[\sigma^1_i] + \Pr[\sigma^3_i]$, we know that $p_i$ does not decrease. Therefore, the solution $\{x_i(\tau)\}_{i=1}^n$  remains feasible for \cref{eq:nonlinear} with the variables $\{p_i\}_{i=1}^n$ for the new mapping, and the right-hand side in \cref{eq:obj} does not decrease as well, since $Q = \prod_{i=1}^n Q_i$ does not decrease. Therefore, the utility $u_i(\tau)$ does not decrease. \Cref{fig:reduction} illustrates the reduction process for the two different cases discussed above.

For any agent $i$, the above process yields two signals $\sigma^{a}_i, \sigma^b_i$, where $\mu_i(\sigma^a_i) \in \hat{I}$. Recall that $\mu_i = \E[D_i]$ is the prior mean, and is a convex combination of $\mu_i(\sigma^a_i)$ and $\mu_i(\sigma^b_i)$. There are three cases now.
\begin{enumerate}
\item If $\mu_i \in \hat{I}$, then not revealing anything yields $Q_i = p_i = 1$. Keeping $x_i(\tau)$  the same, this preserves feasibility of \cref{eq:nonlinear}. Since the right-hand side of \cref{eq:obj} does not decrease, the utility $u_i(\tau)$ does not decrease. This is also the mapping $\sigma^a_i$ that maximizes $\Pr[\sigma^a_i]$. 
\item If $\mu_i < m$, then $\mu_i(\sigma^b_i) < m$, which means $Q_i = 1$, and $p_i = \Pr[\sigma^a_i]$. If we use the mapping that maximizes $\Pr[\sigma^a_i]$ and keeping $x_i(\tau)$ the same, then this preserves the feasibility of \cref{eq:nonlinear} and does not decrease the RHS of \cref{eq:obj}, hence showing the utility $u_i(\tau)$ does not decrease. %\yh{I think we may need some more arguments here, using the polymatroid constraints. Since Eq. (5) and Eq. (6) are not sufficient for a utility vector.(?)}
\item If $\mu_i > \hat{m}$, then  $\mu_i(\sigma^b_i) > \hat{m}$. In this case, $p_i = 1$, and $Q_i = \Pr[\sigma^a_i]$. Again, using the mapping that maximizes $\Pr[\sigma^a_i]$ and keeping $x_i(\tau)$ the same, this preserves the feasiblity of \cref{eq:nonlinear} and does not decrease the RHS of \cref{eq:obj}, showing the utility $u_i(\tau)$ does not decrease.
\end{enumerate}

Therefore, the resulting solution uses a maximal mapping and preserves \cref{eq:obj,eq:nonlinear}.% and the utilities.
\end{proof}

\begin{figure}[htbp]
    \centering
    \includegraphics[scale = 0.8] {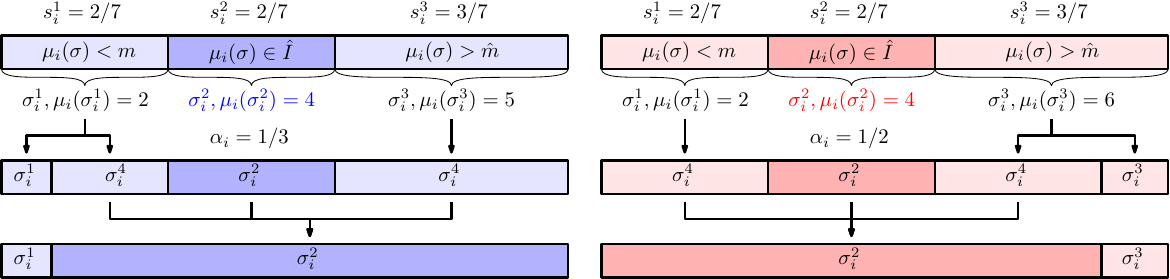}
    \caption{Illustration of the reduction process for two instances.  In both instances, we have $s_i^1 = 2/7$, $s_i^2 = 2/7$, $s_i^3 = 3/7$; with $\mu_i(\sigma_i^1) = 2$ and $\mu_i(\sigma_i^2) = 4$. In the blue (left) instance, $\mu_i(\sigma_i^3) = 5$, so $\alpha_i = 1/3$. We absorb all the probability mass from $\sigma_i^3$ (i.e., $3/7$) together with $3/14$ probability of $\sigma_i^1$ into $\sigma_i^2$. In the red (right) instance, $\mu_i(\sigma_i^3) = 6$ thus $\alpha_i = 1/2$. We absorb all the probability mass of $\sigma_i^1$ (i.e., $2/7$) with the same amount of probability mass in $\sigma_i^3$ into $\sigma_i^2$. }
    \label{fig:reduction} 
\end{figure}

\subsection{Reduction to Network Flow and Approximate Majorization} 
For any agent $i$, the maximal mapping (that maximizes $\Pr[\sigma^a_i]$) is the solution to a linear program. Let $y_v$ denote the probability that the value $v$ maps to a signal in $\hat{I}$, and we have the following linear program.
\begin{equation}
\label[Program]{prog:pQ}
\begin{alignedat}{3}
    & \text{maximize: } & \sum_v y_v & \\
    & \text{subject to: } \quad & \sum_v v \cdot y_v & \geq m \cdot \sum_v y_v; \\
     && \sum_v v \cdot y_v & \leq \hat{m} \cdot \sum_v y_v; \\
     && 0 \leq y_v & \leq \Pr[v], \qquad & \forall v.
\end{alignedat}
\end{equation}
%This shows that we can computationally efficiently find the maximal mapping.

%Now we move on to prove our main approximate majorization result.
Let the optimal objective value of \cref{prog:pQ} be $\beta_i$. If $\mu_i < m$, then we set $Q_i = 1$ and $p_i = \beta_i$; if $\mu_i > \hat{m}$, then we set $p_i = 1$ and $Q_i = \beta_i$; else we set $p_i = Q_i = 1$. Again, let $Q = \prod_i Q_i$.

We now rewrite the constraints and objective for the \iss{} $\tau$ as follows. Note that $\{p_i\}_{i = 1}^n$ and $Q$ are calculated as above, and do not depend on the policy. 

\begin{equation}
\label[Program]{prog:pmid0}
\begin{alignedat}{3}
     m \cdot Q \cdot x_i(\tau) & =  u_i(\tau), & \qquad & \forall i \in [n]; \\
     \sum_{i \in S} x_i(\tau) & \leq 1 - \prod_{i \in S} (1- p_i(\tau))  & \qquad & \forall S \subseteq [n]. \\
%     \sum_{i = 1}^n x_i(\tau) &\leq 1; \\
%     0 \le x_i(\tau) &\leq p_i, & \qquad & \forall i \in [n]. \\
\end{alignedat}
\end{equation}

%\[ \begin{array}{rcll} 
%      \sum_{i \in T} x_i(\tau) & \le  & 1 - \prod_{i \in T} (1 - p_i) & \forall T \subseteq [n] \\
%    x_i & \ge & 0 & \forall i
%\end{array} \]

Note that $\Omega_k^* = \sum_{\ell \ge 1} \rho_{\ell} \cdot \tau_{\ell}$, and after our transformation above, each $\tau_{\ell}$ now faces the same constraints in \cref{prog:pmid0}. Therefore, we can set $x_i = \sum _{\ell \ge 1} \rho_{\ell} \cdot x_i(\tau_{\ell})$, which will be feasible in \cref{prog:pmid0}. %The constraints now define a network flow instance in the same fashion as that in \cref{sec:network_flow}. 
In order to reduce to network flow, we will relax the above program to the following:

\begin{equation}
\label[Program]{prog:pmid}
\begin{alignedat}{3}
     m \cdot Q \cdot x_i(\tau) & =  u_i(\tau), & \qquad & \forall i \in [n]; \\
     \sum_{i = 1}^n x_i(\tau) &\leq 1; \\
     0 \le x_i(\tau) &\leq p_i, & \qquad & \forall i \in [n]. \\
\end{alignedat}
\end{equation}

Recall from \cref{sec:reduction} that we randomize over the projections $\{\Omega^*_k\}_{k = 1}^K$ of the optimal \sgp{} $\Omega^*$, while preserving all utilities to a factor of $K$. For every $k \in \{1,2,\ldots,K\}$, consider the relaxed \sgm{} policy constructed above. For notational convenience, we denote the $p_i, Q_i$ in the corresponding maximal mapping as $p^k_i, Q^k_i$. Further denote the lower end-point of $I_k$ as $m^k$, and the variable $x_i$ by $x_i^k$. Putting together \cref{prog:pmid} for different values of $k$, it holds that the utilities $\{u_i\}_{i = 1}^n$ of the policy $\Omega^*$, after losing the factor of $K$, satisfy the following constraints.
\begin{equation}
\label[Program]{prog:pmaj}
\begin{alignedat}{3}
    \frac{1}{K} \cdot \sum_{k=1}^K m^k \cdot Q^k \cdot x^k_i & =  u_i, & \qquad & \forall i \in [n]; \\
    \sum_{i =1}^n x^k_i &\leq 1, & \qquad & \forall k \in [K];\\
    0 \leq x^k_i &\leq p^k_i, & \qquad & \forall i \in [n], k \in [K].
\end{alignedat}
\end{equation}

%\begin{equation} 
%\label{eq:convex_prog}
%\mbox{Maximize } \ \  f\left( \left \{ \frac{1}{K} \sum_{k=1}^K m^k \cdot Q^k \cdot  x^k_i \right\}_{i=1}^n \right) 
%\end{equation}
%subject to
%\[ \begin{array}{rcll} 
%    u_i & \le & \frac{1}{K} \sum_{k=1}^K m^k \cdot Q^k \cdot  x^k_i & \forall i \in [n] \\
%    \sum_{i \in T} x^k_i & \le  & 1 - \prod_{i \in T} (1 - p^k_{i}) & \forall T \subseteq [n], k \in [K]\\
%    x^k_i & \ge & 0 & \forall i
%\end{array} \]

%\paragraph{Network Flow.} We next relax the constraints in the above program as follows.

%\kw{What is $T$ above?}

\paragraph{Network Flow and Majorization.} Define $b^k := m^k \cdot Q^k / K$ and $y_i^k := b^k \cdot x^i_k$, and we have
\begin{equation}
\label[Program]{prog:pmaj3}
\begin{alignedat}{2}
\sum_{k=1}^K y^i_k & =  u_i, & \qquad & \forall i \in [n]; \\
\sum_{i=1}^n y^k_i &\leq b^k, & \qquad & \forall k \in [K]; \\
0 \leq y^k_i &\leq p^k_i \cdot b^k, & \qquad & \forall i \in [n], k \in [K].
\end{alignedat}
\end{equation}
By the same argument as in \cref{sec:full_maj}, the constraints above formulate a network flow problem. The agents are the sinks. Each interval $\hat{I}_k$ is an intermediate node that is connected to the source with edges of capacity $b^k$. The node $\hat{I}_k$ is connected to sink $t_i$ with an edge of capacity $p_i^k \cdot b^k$. %\yh{The structure of the network is illustrated in \cref{fig:network_3}.} 
By \cref{thm:bern_main}, this network flow problem has a $1$-majorized solution (which is computable in polynomial time according to the discussion in \cref{sec:network_flow}). Consequently, there exist utilities $\{\hat{u}_i\}_{i = 1}^n$ that are $1$-majorized. 

%\begin{figure}
%    \centering
%    \includegraphics[scale = 1.2] {figures/network_3.pdf}
%    \caption{Illustration of the network flow instance corresponding to the constraints on $y_{i}^k$'s. The text in blue denotes the capacity of each edge. }
%    \label{fig:network_3}
% \end{figure}

\paragraph{Final Algorithm.} Our final algorithm will round the majorized solution analogous to \cref{sec:full_maj}.  %The following rounding algorithm will now complete the proof of \cref{thm:approx}. %We now pick a $k \in [K]$ uniformly at random and round the solution corresponding to the interval $\hat{I}_k$ via the same technique as \cref{sec:full_maj}.
%This shows \cref{thm:approx} via  the following algorithm, which is bicriteria $ O(K)$-majorized, and is computable in polynomial time:
\begin{itemize}
\item For each $k \in \{1,2,\ldots,K\}$, consider the relaxed interval $\hat{I}_k = [m_k, \eta \cdot m_k)$, and compute the quantities $\{p^k_i,Q^k_i\}$ using \cref{prog:pQ}.
\item Compute the $1$-majorized solution $\hat{x}$ to  \cref{prog:pmaj}. 
\item Pick a number $k \in \{1,2,\ldots,K\}$ uniformly at random.
\item Given the realized values, compute the corresponding signals via the maximal mapping for $\hat{I}_k$ and send the signals to the receiver. The receiver selects an agent as follows:
\begin{itemize}
\item If any agent has posterior mean larger than the interval $\hat{I}_k$, select the agent with the highest mean and stop.
\item Consider the agents whose posterior mean lies in $\hat{I}_k$ in a uniformly random order. For each agent $i$ in this order, if it is reached, select it with probability $\hat{x}^k_i / p^k_i$ and stop.
\item If no agent is selected so far, select the agent with the highest posterior mean.
\end{itemize}
\end{itemize}

\paragraph{Analysis.} 
The same analysis as \cref{sec:full_maj} shows that for each agent $i$, the utility in the above algorithm is at least $\hat{u}_i / 2$, where $\hat{u}_i$ is the utility in the $1$-majorized solution to \cref{prog:pmaj}. To see this, note that the probability of selecting the interval $K$ is $1 / K$. Conditioned on this event, the probability that no agent has posterior mean larger than $\hat{I}_k$ is $Q^k$. Conditioned on these two events, the probability that the posterior mean of agent $i$ both lies in $\hat{I}_k$ and is selected is at least $\hat{x}^k_i / 2$, and in this case, it yields utility at least $m^k$. The overall utility bound now follows by linearity of expectation. %\km{Should we write out the analysis in more detail, or is this ok?}

Randomizing over \sgm{} projections (see \cref{sec:reduction}) loses a factor of $O(K) = O((\log V) / \epsilon)$ in each utility, and hence our method gives a bicriteria $O((\log V) / \epsilon)$-majorized solution, completing the proof of \cref{thm:approx}.

\subsection{Extension to Asymmetric Utilities}
\label{sec:asym}
So far, we have assumed that the receiver's selection criterion is based on the agents' signaled values, and the value of the chosen agent is the welfare the system generates. The work of~\cite{AuKawai,Du2024} considers a slightly different model, where though the agents signal their values and the receiver selects the agent whose mean posterior signaled value is largest, an individual agent's utility is $1$ if selected and $0$ otherwise. This corresponds to the utility perceived by the selfish agent, as opposed to the welfare generated by the system. We observe that the proofs for \cref{thm:approx,thm:full3} extends as is to this model, and sketch the changes needed. 

To see that the proof of \cref{thm:approx} extend as is, we keep the construction of the mathematical program in  \cref{sec:structure} the same, except we modify \cref{eq:obj} to
\begin{equation}
\label{eq:obj10}
    u_i(\tau) = Q(\tau) \cdot  x_i(\tau), \qquad \forall i \in [n].
\end{equation}
which captures the utility being the same as the probability of being selected. For this modification, the proof of \cref{lem:struct} remains unchanged. Subsequently, \cref{prog:pQ} remains unchanged, while we change \cref{prog:pmaj} to

\begin{equation}
\label[Program]{prog:pmaj2}
\begin{alignedat}{3}
    \frac{1}{K} \cdot \sum_{k=1}^K  Q^k \cdot x^k_i & =  u_i, & \qquad & \forall i \in [n]; \\
    \sum_{i =1}^n x^k_i &\leq 1, & \qquad & \forall k \in [K];\\
    0 \leq x^k_i &\leq p^k_i, & \qquad & \forall i \in [n], k \in [K].
\end{alignedat}
\end{equation}

This again reduces to network flow by defining $b^k := Q^k / K$ and $y_i^k := b^k \cdot x^i_k$, yielding exactly \cref{prog:pmaj3}. This completes the proof of \cref{thm:approx} for this setting. 

Analogously, the proof of \cref{thm:full3} extends to this setting by changing the relaxation to replace the value $v_j$ by $1$ in the first (utility) constraint \cref{eqn:constraint_1}.  It is easy to check that the resulting program this also reduces to network flow. The rest of the analysis goes through as is, except for replacing $v_j$ by $1$ in \cref{eqn:utility1}. This shows a $2$-majorized solution.
\section{Lower Bound on Majorization\texorpdfstring{: Proof of \cref{thm:main_lb}}{}}
\label{sec:main_lb}
%\begin{proposition} \label{prop:loglog_lb}
%    There does not exist a $(\frac{1}{2} \cdot \log\log n)$-majorized signaling scheme.
%\end{proposition}

%\km{We need to change this instance to be on $[1,V]$ and not have $0$ in its support.}

We first assume that the receiver is maximizing the utilitarian welfare exactly, and it will become clear that the same lower bound holds if the receiver selects an agent using an \emph{arbitrary} function of the posterior means -- in particular, if it is a $(1+\epsilon)$-approximate utilitarian welfare maximizer.

Let there be $n$ agents. Each agent $i$ has value $i + 1$ with probability $1/i$ and has value $1$ with probability $1-1/i$. Note that the expected value of agent $i$ is $\mu_i = 2$, the same for each $i$. For any $k \in [n]$, denote $\sum_{i = k} ^ n 1 / (i + 1)$ by $H_{-k}$. We first give a lower bound on the sum of the smallest $k$ agents' utilities.

\begin{lemma} \label{lem:pfx_sum_lb}
    For any $k \in [n]$, there exists a \sgp{} $S_k$, such that if the utilities of the agents are arranged in ascending order, the prefix of the first $k$ utilities sums to \[
    \frac{1}{2} \cdot \frac{k \cdot (k + 3)}{k + (k + 1) \cdot H_{-(k + 1)}}.
    \]
\end{lemma}

In order to prove this lemma, we construct the signaling policy $S_k$ as follows. For each agent $i$ with $i \in \{k+1, \ldots, n\}$, if $v_i = i + 1$, then with probability $x_i$, the sender sends a signal $s_i$ to the receiver, and otherwise the sender sends a signal $\bar s_i$ for this agent. For $i \le k$, the sender reveals the true value of agent $i$. Note that this \mpr{} is independent for each agent.

Note that conditioned on receiving signal $s_i$, we have $\mu_i(s_i) = i + 1$.  On the other hand, conditioned on receiving $\bar s_i$, we have 
\[
\mu_i(\bar s_i) = \frac{\left(1 - \frac{1}{i}\right) \cdot 1 + \frac{1}{i} \cdot (1 - x_i) \cdot (i + 1)}{1 - \frac{1}{i} \cdot x_i} 
= \frac{2 - x_i - \frac{x_i}{i}}{1 - \frac{x_i}{i}} \leq 2.
\]
Therefore, if at least one signal $s_i$ is received, the receiver will select the agent $i$, $i \in [k+1,n]$ with the largest index. On the other hand, if no signal $s_i$ is received, the receiver selects the unique agent $i \in [k]$ with the highest revealed value.

For each $i \in \{k + 1, \ldots, n\}$, we set
\[
x_i = \frac{\frac{i}{i + 1}}{\frac{k}{k + 1} + H_{-(k+1)} - H_{-(i+1)}}.
\]
The intuition behind this formula is to have $U_i = U_{i+1}$ for any $k\le i \le n-1$.  We obtain the utility of all agents as follows.
\begin{lemma} \label{lem:sp_utilities}
It holds that
    \begin{equation*}
    U_j(S_k) = \begin{cases}
        \frac{j + 1}{k + 1} \cdot \frac{1}{\frac{k}{k + 1} + H_{-(k + 1)}} & j \le k;\\
        \frac{1}{\frac{k}{k + 1} + H_{-(k + 1)}} & j \ge k + 1.
    \end{cases} 
    \end{equation*}
\end{lemma}
\begin{proof}
   By the construction of $S_k$, for any $k + 1 \le j\le n$, the receiver will receive signal $s_j$ and no signal $s_i$ for $i > j$ with probability 
    \[
    x_j \cdot \frac{1}{j} \cdot \prod_{i=j+1}^{n} \left(1-\frac{x_i}{i} \right).
    \]
    Since agent $j$ receives utility $j + 1$ if and only if signal $s_j$ is sent and no signal $s_i$ for $i > j$ is sent, we have 
    \begin{align*}
    U_j(S_k) & =  (j + 1) \cdot x_j \cdot \frac{1}{j} \cdot \prod_{i = j + 1}^{n} \left(1-\frac{x_i}{i} \right) = \frac{j + 1}{j} \cdot x_j \cdot \prod_{i = j + 1}^{n} \left(1 - \frac{\frac{1}{i + 1}}{\frac{k}{k + 1} + H_{-(k + 1)} - H_{-(i + 1)}}\right) \\
    & = \frac{1}{\frac{k}{k + 1} + H_{-(k + 1)} - H_{-(j + 1)}} \cdot \prod_{i = j + 1}^{n} \left( \frac{\frac{k}{k + 1} + H_{-(k + 1)} - H_{-(i + 1)} - \frac{1}{i + 1}}{\frac{k}{k + 1} + H_{-(k + 1)} - H_{-(i + 1)}} \right)\\
    & = \frac{1}{\frac{k}{k + 1} + H_{-(k + 1)} - H_{-(j + 1)}} \cdot \prod_{i = j + 1}^{n} \left( \frac{\frac{k}{k + 1} + H_{-(k + 1)} - H_{-i}}{\frac{k}{k + 1} + H_{-(k + 1)} - H_{-(i + 1)}} \right) = \frac{1}{\frac{k}{k + 1} + H_{-(k + 1)}}.
    \end{align*}
    Next we consider agent $j$ for $j \le k$. By a similar telescoping as above, the probability that no signal $s_i$ with $i > k$ is sent (denoted by $q_R$) is 
    \[
        q_R = \prod_{i=k + 1}^{n} \left(1 - \frac{\frac{1}{i + 1}}{\frac{k}{k + 1} + H_{-(k + 1)} - H_{-(i + 1)}} \right) =  \frac{\frac{k}{k + 1} + H_{-(k + 1)} - H_{-(k + 1)}}{\frac{k}{k + 1} + H_{-(k + 1)}} = \frac{\frac{k}{k + 1}}{\frac{k}{k + 1} + H_{-{(k + 1)}}}.
    \]
    For any agent $j$ with $j \leq k$, the probability that it has the largest revealed value among agents $\{1, 2, \ldots, k\}$ is
    \[
    \frac{1}{j} \cdot \prod_{i = j + 1}^{k} \left(1 - \frac{1}{i}\right) = \frac{1}{k}.
    \]
    Therefore, the expected utility of agent $j$ is
    \[
    U_j(S_k) = \frac{j + 1}{k} \cdot q_R = \frac{j + 1}{k} \cdot \frac{\frac{k}{k + 1}}{\frac{k}{k + 1} + H_{-(k + 1)}} = \frac{j + 1}{k + 1} \cdot \frac{1}{\frac{k}{k + 1} + H_{-(k + 1)}}. \qedhere
    \]
\end{proof}

\begin{proof} [Completing the proof of \cref{lem:pfx_sum_lb}]
    We apply the signaling scheme $S_k$ defined above. By \cref{lem:sp_utilities}, we have $U_i(S_k) \le U_j(S_k)$ for all $1\le i\le j \le n$. Thus the $k$ agents with smallest expected utilities are agents $1, 2,\ldots, k$. Their utilities sum up to
    \[
    \sum_{i=1}^k U_i(S_k) = \frac{1}{2} \cdot \frac{k \cdot (k + 3)}{k + (k + 1) \cdot H_{-(k + 1)}},
    \]
    completing the proof.
\end{proof}

\begin{proof} [Proof of \cref{thm:main_lb}]
    Since $k \ge 1$, we have
    \[
    \frac{1}{2} \cdot \frac{k \cdot (k + 3)}{k + (k + 1) \cdot H_{-(k + 1)}} > \frac{1}{2} \cdot \frac{k \cdot (k + 3)}{(k + 1) + (k + 1) \cdot H_{-(k + 1)}} \ge  \frac{1}{2} \cdot \frac{k + 1}{1 + H_{-(k + 1)}} =: R_k.
    \]
    The quantity $R_k$ is a lower bound on the sum of the $k$ smallest utilities in the \sgp{} $S_k$. Fix any \sgp{} $S'$. Suppose that each agent $i$ is selected with probability $w_i$ in $S'$. If $S'$ is $\alpha$-majorized, the sum of the utilities of agents in the set $\{1, 2, \ldots, k\}$ must be at least $R_k / \alpha$. Therefore, we have the following set of inequalities
    \begin{equation*}
       \sum_{i=1}^k (i + 1)\cdot w_i \ge \frac{1}{\alpha} \cdot R_k, \qquad \forall k \in \{1,2, \ldots, n\}.
       %\\
       % w_1 \cdot 1 + w_2 \cdot 2 \ge \frac{1}{\alpha} \cdot R_2;\\
       % \cdots\\
       % w_1 \cdot 1 + w_2 \cdot 2 + \cdots + w_n \cdot n \ge \frac{1}{\alpha} \cdot R_n.
    \end{equation*}

    We now want to lower bound $\sum_{i=1}^n w_i$. The optimal way to assign the probabilities $\{w_i\}_{i = 1}^n$ is to make all the inequalities satisfied with equality. Therefore, we have
    \begin{align*}
        \sum_{i = 1}^n w_i &\ge \frac{1}{2\alpha} \cdot R_1 + \frac{1}{\alpha} \cdot \sum_{i = 2}^n \left(\frac{R_{i} - R_{i - 1}}{i + 1} \right)\\
        & \ge \frac{1}{\alpha} \cdot \sum_{i = 1}^n \frac{R_i}{(i + 1) \cdot (i + 2)} = \frac{1}{2\alpha} \cdot \sum_{i = 1}^n \frac{1}{(i + 2) \cdot (1 + H_{-(i + 1)})} \\ %\tag{Plugging in $R_i$} \\
        & \ge \frac{1}{2\alpha} \cdot \sum_{i = 1}^n \frac{1}{(i + 2)\cdot (1 + \log (n + 1)- \log (i + 1))} \tag{Since $H_{-(i + 1)} \le \log \left(\frac{n + 1}{i + 1}\right)$} \\
        & \ge \frac{1}{2\alpha} \cdot \sum_{i = 1}^n \frac{1}{(\frac{3}{2} i + \frac{3}{2}) \cdot (1 + \log (n + 1) - \log (i + 1))}\\
        & \ge \frac{1}{3\alpha} \cdot \int_{1}^{n+1} \frac{1}{x\cdot(1 + \log (n+1) - \log x)} \d x \tag{Since $\frac{1}{x\cdot (1 + \log(n + 1) - \log x)}$ is monotonically decreasing on [1, n+1]}\\
        & = \frac{1}{3\alpha} \cdot \left[ -\log(1+\log(1 + n) - \log x)\Big\vert_{x = 1}^{n + 1} \right]  = \frac{1}{3\alpha} \cdot \log(1 + \log(1 + n)). 
    \end{align*} 
    Combining the inequality above with $\sum_{i=1}^n w_i \le 1$, we have
    \[
    \alpha \ge \frac{\log{(1 + \log {(1 + n)})}}{3} > \frac{1}{3} \cdot \log \log (n + 1) = \frac{1}{3} \cdot \log \log V. \qedhere
    \]
\end{proof}

Since the second part of our argument only requires the fact that at most one agent is selected and is independent of the receiver's \slr{}, this lower bound also applies if the receiver can arbitrarily select the agent, or (in particular) if the receiver is an approximate utilitarian welfare maximizer.

\section{Open Directions}
Our work points to several interesting future directions. First, we assume independent (or decentralized) mapping of values to signals, motivated hiring and selection applications. Can a similar existence result as in \cref{sec:approx} extend to the setting where the sender can correlate signals from different agents? We note however that the two settings are somewhat incomparable, and it may very well be that the correlated setting is computationally simpler for a given fairness function~\cite{dughmi2016algorithmic}, while the independent case is easier from an approximate majorization perspective. 

Secondly, our lower bound in \cref{sec:main_lb} holds for distributions with large variance. Is there an $O(1)$-majorized policy under a more benign assumption on distributions, such as the monotone hazard rate (MHR) assumption?   Next, we assumed a single agent is finally selected. What if the receiver selects the top $k$ agents according to the posterior mean?  In this case, it is open if there a polynomial-time algorithm for a given fairness function, in both the cases where the sender can correlate signals and when it sends independent signals. Further, it is an open question to extend our majorization result even to the case when $k = 2$.

Finally, we assume agents are not strategic in revealing information, and our results can be viewed as the limits of fairness that is achievable even if agents follow a prescribed policy. It is an interesting direction to study equilibria and price of anarchy when agents reveal information strategically, building on \cite{AuKawai,Du2024}.

%\newpage
\bibliographystyle{alpha}
\bibliography{refs}

\newcommand{\etalchar}[1]{$^{#1}$}
\begin{thebibliography}{CDHW20}

\bibitem[AK20]{AuKawai}
Pak~Hung Au and Keiichi Kawai.
\newblock Competitive information disclosure by multiple senders.
\newblock {\em Games and Economic Behavior}, 119:56--78, 2020.

\bibitem[BBM15]{bergemann2015limits}
Dirk Bergemann, Benjamin Brooks, and Stephen Morris.
\newblock The limits of price discrimination.
\newblock {\em American Economic Review}, 105(3):921--957, 2015.

\bibitem[BM04]{bertrand2004emily}
Marianne Bertrand and Sendhil Mullainathan.
\newblock Are {E}mily and {G}reg more employable than {L}akisha and {J}amal?
  {A} field experiment on labor market discrimination.
\newblock {\em American Economic Review}, 94(4):991--1013, 2004.

\bibitem[BM19]{bergemann2019information}
Dirk Bergemann and Stephen Morris.
\newblock Information design: A unified perspective.
\newblock {\em Journal of Economic Literature}, 57(1):44--95, 2019.

\bibitem[BMSW24]{Banerjee2024fair}
Siddhartha Banerjee, Kamesh Munagala, Yiheng Shen, and Kangning Wang.
\newblock Fair price discrimination.
\newblock In {\em Proceedings of the 2024 {ACM-SIAM} Symposium on Discrete
  Algorithms (SODA)}, pages 2679--2703. {SIAM}, 2024.

\bibitem[BTXZ21]{DBLP:conf/sigecom/BabichenkoTXZ21}
Yakov Babichenko, Inbal Talgam{-}Cohen, Haifeng Xu, and Konstantin Zabarnyi.
\newblock Regret-minimizing bayesian persuasion.
\newblock In {\em Proceedings of the 22nd {ACM} Conference on Economics and
  Computation (EC)}, page 128. {ACM}, 2021.

\bibitem[CDHW20]{cummings2020algorithmic}
Rachel Cummings, Nikhil~R. Devanur, Zhiyi Huang, and Xiangning Wang.
\newblock Algorithmic price discrimination.
\newblock In {\em Proceedings of the 2020 {ACM-SIAM} Symposium on Discrete
  Algorithms (SODA)}, pages 2432--2451. {SIAM}, 2020.

\bibitem[CDW12]{cai2012optimal}
Yang Cai, Constantinos Daskalakis, and S.~Matthew Weinberg.
\newblock Optimal multi-dimensional mechanism design: Reducing revenue to
  welfare maximization.
\newblock In {\em Proceedings of the 53rd Annual {IEEE} Symposium on
  Foundations of Computer Science (FOCS)}, pages 130--139. {IEEE} Computer
  Society, 2012.

\bibitem[CH14]{DBLP:journals/mktsci/ChakrabortyH14}
Archishman Chakraborty and Rick Harbaugh.
\newblock Persuasive puffery.
\newblock {\em Marketing Science}, 33(3):382--400, 2014.

\bibitem[CMV20]{celis2020interventions}
L.~Elisa Celis, Anay Mehrotra, and Nisheeth~K. Vishnoi.
\newblock Interventions for ranking in the presence of implicit bias.
\newblock In {\em Proceedings of the 2020 ACM Conference on Fairness,
  Accountability, and Transparency (FAT*)}, pages 369--380. {ACM}, 2020.

\bibitem[CS19]{chakrabarty2019approximation}
Deeparnab Chakrabarty and Chaitanya Swamy.
\newblock Approximation algorithms for minimum norm and ordered optimization
  problems.
\newblock In {\em Proceedings of the 51st Annual {ACM} {SIGACT} Symposium on
  Theory of Computing (STOC)}, pages 126--137. {ACM}, 2019.

\bibitem[DGV08]{DBLP:journals/mor/DeanGV08}
Brian~C. Dean, Michel~X. Goemans, and Jan Vondr{\'{a}}k.
\newblock Approximating the stochastic knapsack problem: The benefit of
  adaptivity.
\newblock {\em Mathematics of Operations Research}, 33(4):945--964, 2008.

\bibitem[DKKS24]{devic2024stability}
Siddartha Devic, Aleksandra Korolova, David Kempe, and Vatsal Sharan.
\newblock Stability and multigroup fairness in ranking with uncertain
  predictions.
\newblock In {\em Proceedings of the 41st International Conference on Machine
  Learning (ICML)}. {PMLR}, 2024.

\bibitem[DKQ16]{dughmi2016persuasion}
Shaddin Dughmi, David Kempe, and Ruixin Qiang.
\newblock Persuasion with limited communication.
\newblock In {\em Proceedings of the 17th {ACM} Conference on Economics and
  Computation (EC)}, pages 663--680. {ACM}, 2016.

\bibitem[DR89]{dutta1989concept}
Bhaskar Dutta and Debraj Ray.
\newblock A concept of egalitarianism under participation constraints.
\newblock {\em Econometrica}, 57(3):615--635, 1989.

\bibitem[DREM13]{dixon2013race}
Ezekiel~J. Dixon-Rom{\'a}n, Howard~T. Everson, and John~J. McArdle.
\newblock Race, poverty and {SAT} scores: Modeling the influences of family
  income on black and white high school students’ {SAT} performance.
\newblock {\em Teachers College Record}, 115(4):1--33, 2013.

\bibitem[DTWZ24]{Du2024}
Zhicheng Du, Wei Tang, Zihe Wang, and Shuo Zhang.
\newblock Competitive information design with asymmetric senders.
\newblock In {\em Proceedings of the 25th {ACM} Conference on Economics and
  Computation (EC))}, 2024.

\bibitem[Dug17]{dughmi2017algorithmic}
Shaddin Dughmi.
\newblock Algorithmic information structure design: a survey.
\newblock {\em SIGecom Exchanges}, 15(2):2--24, 2017.

\bibitem[DX16]{dughmi2016algorithmic}
Shaddin Dughmi and Haifeng Xu.
\newblock Algorithmic bayesian persuasion.
\newblock In Daniel Wichs and Yishay Mansour, editors, {\em Proceedings of the
  48th Annual {ACM} {SIGACT} Symposium on Theory of Computing (STOC)}, pages
  412--425. {ACM}, 2016.

\bibitem[EGGL20]{emelianov2020fair}
Vitalii Emelianov, Nicolas Gast, Krishna~P. Gummadi, and Patrick Loiseau.
\newblock On fair selection in the presence of implicit variance.
\newblock In {\em Proceedings of the 21st {ACM} Conference on Economics and
  Computation (EC)}, pages 649--675. {ACM}, 2020.

\bibitem[GB21]{garciasoriano2021maxmin}
David Garc{\'{\i}}a{-}Soriano and Francesco Bonchi.
\newblock Maxmin-fair ranking: Individual fairness under group-fairness
  constraints.
\newblock In {\em Proceedings of 27th {ACM} {SIGKDD} Conference on Knowledge
  Discovery and Data Mining (KDD)}, pages 436--446. {ACM}, 2021.

\bibitem[GM06]{goel2006simultaneous}
Ashish Goel and Adam Meyerson.
\newblock Simultaneous optimization via approximate majorization for concave
  profits or convex costs.
\newblock {\em Algorithmica}, 44(4):301--323, 2006.

\bibitem[GMP05]{goel2005approximate}
Ashish Goel, Adam Meyerson, and Serge~A. Plotkin.
\newblock Approximate majorization and fair online load balancing.
\newblock {\em ACM Transactions on Algorithms}, 1(2):338--349, 2005.

\bibitem[GMS10]{DBLP:journals/jacm/GuhaMS10}
Sudipto Guha, Kamesh Munagala, and Peng Shi.
\newblock Approximation algorithms for restless bandit problems.
\newblock {\em Journal of the ACM}, 58(1):3:1--3:50, 2010.

\bibitem[GR00]{goldin2000orchestrating}
Claudia Goldin and Cecilia Rouse.
\newblock Orchestrating impartiality: The impact of ``blind'' auditions on
  female musicians.
\newblock {\em American Economic Review}, 90(4):715--741, 2000.

\bibitem[HLP34]{hardy1934inequalities}
Godfrey~Harold Hardy, John~Edensor Littlewood, and George P{\'o}lya.
\newblock {\em Inequalities}.
\newblock Cambridge university press, 1934.

\bibitem[Kar32]{karamata1932inegalite}
Jovan Karamata.
\newblock Sur une in{\'e}galit{\'e} relative aux fonctions convexes.
\newblock {\em Publications de l'Institut Mathematique}, 1(1):145--147, 1932.

\bibitem[KG11]{Kamenica2011bayesian}
Emir Kamenica and Matthew Gentzkow.
\newblock Bayesian persuasion.
\newblock {\em American Economic Review}, 101(6):2590--2615, 2011.

\bibitem[KK06]{kumar2006fairness}
Amit Kumar and Jon~M. Kleinberg.
\newblock Fairness measures for resource allocation.
\newblock {\em SIAM Journal on Computing}, 36(3):657--680, 2006.

\bibitem[KR18]{kleinberg2018selection}
Jon~M. Kleinberg and Manish Raghavan.
\newblock Selection problems in the presence of implicit bias.
\newblock In {\em Proceedings of the 9th Conference on Innovations in
  Theoretical Computer Science Conference (ITCS)}, pages 33:1--33:17. Schloss
  Dagstuhl - Leibniz-Zentrum f{\"{u}}r Informatik, 2018.

\bibitem[MC21]{mehrotra2021mitigating}
Anay Mehrotra and L.~Elisa Celis.
\newblock Mitigating bias in set selection with noisy protected attributes.
\newblock In {\em Proceedings of the 2021 {ACM} Conference on Fairness,
  Accountability, and Transparency (FAccT)}, pages 237--248. {ACM}, 2021.

\bibitem[Meg74]{megiddo1974optimal}
Nimrod Megiddo.
\newblock Optimal flows in networks with multiple sources and sinks.
\newblock {\em Mathematical Programming}, 7:97--107, 1974.

\bibitem[SKJ21]{singh2021fairness}
Ashudeep Singh, David Kempe, and Thorsten Joachims.
\newblock Fairness in ranking under uncertainty.
\newblock In {\em Proceedings of the 35th Annual Conference on Neural
  Information Processing Systems (NeurIPS)}, pages 11896--11908, 2021.

\bibitem[SWZ{\etalchar{+}}23]{shen2023fairness}
Zeyu Shen, Zhiyi Wang, Xingyu Zhu, Brandon Fain, and Kamesh Munagala.
\newblock Fairness in the assignment problem with uncertain priorities.
\newblock In {\em Proceedings of the 2023 International Conference on
  Autonomous Agents and Multiagent Systems (AAMAS)}, pages 188--196. {ACM},
  2023.

\bibitem[Tam95]{tamir1995least}
Arie Tamir.
\newblock Least majorized elements and generalized polymatroids.
\newblock {\em Mathematics of Operations Research}, 20(3):583--589, 1995.

\bibitem[Vei71]{veinott1971least}
Arthur~F. Veinott{ }Jr.
\newblock Least $d$-majorized network flows with inventory and statistical
  applications.
\newblock {\em Management Science}, 17(9):547--567, 1971.

\bibitem[XRDT15]{DBLP:conf/aaai/XuRDT15}
Haifeng Xu, Zinovi Rabinovich, Shaddin Dughmi, and Milind Tambe.
\newblock Exploring information asymmetry in two-stage security games.
\newblock In {\em Proceedings of the 29th {AAAI} Conference on Artificial
  Intelligence (AAAI)}, pages 1057--1063. {AAAI} Press, 2015.

\end{thebibliography}

%\appendix

\end{document}